\documentclass[12pt,draftclsnofoot,onecolumn,comsoc]{IEEEtran}
\usepackage{amsmath}
\usepackage{amssymb}
\usepackage{amsthm}
\usepackage{bm}
\usepackage{graphicx}
\usepackage{graphics}
\usepackage{subfigure}
\usepackage{algorithm}
\usepackage{algorithmic}
\usepackage{setspace}
\usepackage[algo2e,linesnumbered,ruled]{algorithm2e}
\usepackage{float}
\usepackage{array}
\usepackage{xcolor}

\newtheorem{definition}{Definition}
\newtheorem{thm}{Theorem}

\newcommand{\FuncName}[1]{\mbox{\normalfont\textsc{#1}}}
\SetKwProg{Fn}{Function}{}{end}
\newcommand{\Int}{\KwSty{int}}

\newcolumntype{M}[1]{>{\centering\arraybackslash}m{#1}}

\DeclareSymbolFont{symbolsC}{U}{txsyc}{m}{n}
\DeclareMathSymbol{\notniFromTxfonts}{\mathrel}{symbolsC}{61}

\ifCLASSOPTIONcompsoc
  \usepackage[nocompress]{cite}
\else
  \usepackage{cite}
\fi
\ifCLASSINFOpdf
\else
\fi
\begin{document}
%
\title{Frequent-Pattern Based Broadcast Scheduling for Conflict Avoidance in Multi-Channel Data Dissemination Systems}

\author{Chuan-Chi Lai, Yu-De Lin, and~Chuan-Ming Liu
\IEEEcompsocitemizethanks{
\IEEEcompsocthanksitem C.-C. Lai is with the Department of Information Engineering and Computer Science, Feng Chia University, Taichung 40724, Taiwan (e-mail: chuanclai@fcu.edu.tw).
\IEEEcompsocthanksitem Y.-D. Lin is with Wistron (Taiwan), Taipei 11469, Taiwan (e-mail: jasons8021@gmail.com).
\IEEEcompsocthanksitem C.-M. Liu is with the Department of Computer Science and Information Engineering, National Taipei University of Technology, Taipei 10618, Taiwan (e-mail: cmliu@ntut.edu.tw).}
}

%
%


\IEEEtitleabstractindextext{%
\begin{abstract}
With the popularity of mobile devices, using the traditional client-server model to handle a large number of requests is very challenging. Wireless data broadcasting can be used to provide services to many users at the same time, so reducing the average access time has become a popular research topic. For example, some location-based services (LBS) consider using multiple channels to disseminate information to reduce access time. However, data conflicts may occur when multiple channels are used, where multiple data items associated with the request are broadcast at about the same time. In this article, we consider the channel switching time and identify the data conflict issue in an on-demand multi-channel dissemination system. We model the considered problem as a Data Broadcast with Conflict Avoidance (DBCA) problem and prove it is NP-complete. We hence propose the frequent-pattern based broadcast scheduling (FPBS), which provides a new variant of the frequent pattern tree, FP*-tree, to schedule the requested data. Using FPBS, the system can avoid data conflicts when assigning data items to time slots in the channels. In the simulation, we discussed two modes of FPBS: online and offline. The results show that, compared with the existing heuristic methods, FPBS can shorten the average access time by 30\%.
\end{abstract}

}

\maketitle

\IEEEdisplaynontitleabstractindextext

%
\IEEEpeerreviewmaketitle

\section{Introduction}\label{sec:introduction}


%
%
%
%
\IEEEPARstart{W}{ith} advances in wireless communications technologies, mobile devices deeply affect our daily lives, such as notebooks, smart phones, and tablets. Users can easily access various information services, such as on-line news, traffic information,and stock prices. Recently, wireless data dissemination becomes a popular topic~\cite{7558129,7287641,1649413}, which can transmit information to a number of users simultaneously. In comparison with the conventional end-to-end transmission (or client-Server) model, wireless data dissemination can make use of wireless network channels to reduce the delivery time for obtaining information. Wireless data broadcasting is well-suited to the Location-Based Services (LBS) in an asymmetric communication environment, where a large number of users are interested in popular information such as news~\cite{8102432}, traffic reports~\cite{1412085}, and multimedia streams~\cite{8260541,8680668}.

In general, wireless data dissemination can be classified into two modes: push-based and pull-based (on-demand). In push-based wireless data dissemination environments~\cite{Viswanathan94adaptivewireless,DBLP:conf/wmcsa/ZdonikFAA94,Acharya1996}, data items are disseminated cyclically according to a predefined schedule. In fact, the access pattern of data items may change dynamically, and the broadcast frequency of popular data items may be lower than the broadcast frequency of unpopular data items. Such a case will result in a poor average access latency.
In view of this, pull-based wireless data dissemination~\cite{811450,Liu201376,He2015118} that disseminates data items timely according to the received requests was proposed to overcome the aforementioned drawback. In the pull-based mode, the users first upload their demand information to the server through the uplink channel, and then the relevant information will be immediately arranged into the broadcasting channels for disseminating data to users. In wireless data dissemination environments, a way of judging the quality of a scheduling approach is to measure the access time of the generated schedule. The access time is a measured time period from starting tuning the channels to obtaining all the requested information. Thus, it is important to have a better broadcasting schedules for shorter access time.

\subsection{Motivation}
In early literature, some conventional works~\cite{7524596,1549811,WU2005714} focus on how to maximize the bandwidth throughout or minimized the access time in single channel environments. Recently, with the advance on antenna techniques, most of works~\cite{1039849,Zheng:2005:TNS:1071246.1071252,4407711} has shifted their focus on the similar issues in multiple channel environments. In general, a multi-channel wireless data dissemination system can provide a more network bandwidth and a shorter access time for data dissemination than a single-channel wireless data dissemination system can. 

However, one new issue, data conflict~\cite{6195838,6216357,6477044}, emerges while each client retrieves data items on multiple channels with channel switching in push-based broadcasting environments.
Two types of conflicts may occur in multi-channel dissemination systems. The first type of conflict is that two required data items are allocated on the same time slot of different channels, so the client cannot download the required data items simultaneously. The second type of conflict occurs if two required data items are allocated on the $t$ and $(t+1)$ time slots of different channels respectively. In such a scenario, the client cannot download both required data items during the time period $[t, t+1]$. The 1st conflict type is obvious. The
reason of the 2nd conflict type is that switching from any channel to a different channel takes time. A client cannot download data at time slot $t+1$ from one channel if it was downloading data item from another channel at time slot $t$, because a time slot is already the smallest unit for data retrieving. Note that a client is allowed to access one channel at one time.

Such a data conflict issue makes a client miss its needed data items during the time period for channel switching, thereby leading to a worse access time. 
On one hand, some works~\cite{6195838,6216357,6477044} provide some solutions from the client's point of view. These solutions can make each client schedule itself for retrieving the data items on channels efficiently. 
On the other hand, only one work~\cite{He2015118} provides a server-side scheduling algorithm with consideration of the data conflict issue in on-demand multi-channel environments. The provided algorithm considers the associations between data items and requests while allocating data items on multiple channels and this provides a conflict-free schedule.

Most broadcast scheduling techniques in on-demand multi-channel data dissemination environments do not consider the time requirement for channel switching, thereby leading to data conflicts or long access time. This phenomenon motivates us to propose a more efficient server-side scheduling method with conflict avoidance using frequent pattern mining technique, thereby shortening the average access time.

\subsection{Contribution}
In this study, we discuss how to shorten the average access time on a multi-channel wireless data dissemination environment under the data conflict conditions. The contributions of this work are listed as follows.
\begin{enumerate}
	\item Identify the data broadcast with conflict avoidance (DBCA) problem in on-demand multi-channel wireless data dissemination environments and prove the considered DBCA problem is $\mathcal{NP}$-complete.
	\item We propose a heuristic approach, Frequent-Pattern based Broadcast Scheduling (FPBS), for providing an approximate schedule in polynomial time. Inspired by frequent-pattern tree (FP-tree), we suggest a new tree, FP*-tree, for FPBS to schedule the requested data items with the consideration of channel switching.
	\item We analyze the time complexity and average access time of FPBS in both average case and worst case.
	\item We verify the performance of FPBS which achieves a shorter average access time in comparison with the existing method, UPF~\cite{He2015118}.
\end{enumerate}

The rest of this paper is organized as follows. Section~\ref{related_work} gives the background and reviews related research in the literature. Section~\ref{problem} defines the DBCA problem and proves that the DBCA problem is $\mathcal{NP}$-complete. Section~\ref{FPBS} explains the proposed approach with examples and algorithms in detail. In Section~\ref{analysis}, we discuss the time complexity and access time of the proposed approach in worst case. Section~\ref{simulation} presents the experimental simulation results and validates the correctness and effectiveness of the proposed methods in various situations. Finally, we conclude this work in Section~\ref{conclusion}.



\section{Related Work}
\label{related_work}
In the multi-channel dissemination environments, many related research works focused on data scheduling to improve the access time performance~\cite{1039849,Zheng:2005:TNS:1071246.1071252} from the perspective of spectrum utilization. Yee et al.~\cite{1039849} proposed a greedy algorithm to find the best way to distribute data items into the channels, allowing users to access requested data in a limited time. Zheng et al.~\cite{Zheng:2005:TNS:1071246.1071252} considered the data access frequency, data length, and channel bandwidth into a model and proposed a Two-level Optimization Scheduling Algorithm (TOSA) to find an appropriate schedule. They also showed that the schedule of TOSA is approximate to the best average time. Yi et al.~\cite{4407711} proposed a method to allow replicating multiple copies of a data item in a broadcasting channel. If there are multiple copies of a popular data item in the channel, the average access time can be effectively reduced.

In addition to the above methods, some works considered the priority of incoming queries and found ways to reduce the access time~\cite{Liu201376,7524596,1549811,LV201223}. Lu et al.~\cite{7524596} proposed some algorithms to schedule data for Maximum Throughput Request Selection (MTRS) and Minimum Latency Request Ordering (MLRO) problems in a single-channel environment and proved that both problems are $\mathcal{NP}$-hard. Xu et al.~\cite{1549811} proposed a $SIN-\alpha$ algorithm with a set of priority decisions based on the ratio of the length of the expiration time over the amount of information. 
Lv et al.~\cite{LV201223} proved that minimizing access time in the broadcasting scheduling of multi-item requests with deadline constraint in a single channel environment is an $\mathcal{NP}$-hard problem. The authors provided a profit-based heuristic scheduling algorithm to minimize the request miss rate (or delivery miss rate) considering the access frequency of data.
Liu and Su~\cite{Liu201376} focused on reducing the demand for the loss rate and shortening the access time. Two kinds of algorithms, Most Popular First Heuristic (MPFH) and Most Popular Last Heuristic (MPLH), were proposed to solve the problems and they also analyzed differences between the online version (the user demands continuously come in the system, so the scheduling task needs to wait until it starts receiving information of the demands) and offline version (the system already has all the information of demands). 

Some works had found that the dependency between requested data items may greatly influence the performance of multi-channel data broadcasting.
Lin and Liu~\cite{Lin:2006:BDD:1143549.1143711} considered the dependencies among data items as a Directed Acyclic Graph (DAG). They proved that finding the best schedule preserving dependencies between each data item is an $\mathcal{NP}$-hard problem and proposed some heuristics for the problem.
Qiu et al.~\cite{8309412} proposed a three-layer on-demand data broadcasting (ODDB) system for enhancing the uplink access capacity by introducing a virtual node layer. Each virtual node can merge duplicated requests and help the server reduce huge computational load, there by improve the broadcasting efficiency. 

\begin{table*}[!t]
	\renewcommand{\arraystretch}{1.1}
	\caption{Comparisons of Related Works and the Proposed Method}
	\label{compared_methods}
	\centering
	\small
	\begin{tabular}{|M{1.5cm}|M{1.8cm}|M{1.8cm}|c|M{1.6cm}|M{1.85cm}|M{1.5cm}|M{1.25cm}|}
		\hline
		\textbf{Related Works} & \textbf{System Model} & \textbf{Criterion} & \textbf{Method}& \textbf{Data Popularity} & \textbf{Request Dependency} & \textbf{Channel Switching} & \textbf{Data Conflict} \\ \hline
		\hline
		\cite{1039849,Zheng:2005:TNS:1071246.1071252,4407711} & On-demand & Latency & Server-side & No & No & No & No \\ \hline	
		\cite{Liu201376}\cite{7524596}\cite{1549811} & On-demand & Latency & Server-side & Yes & No & No & No \\ \hline
		\cite{6216357} & Push-based & Latency & Client-side & No & No & Yes & Yes \\ \hline	
		\cite{6477044} & Push-based & Throughput & Client-side & No & No & Yes & Yes \\ \hline		
		\cite{LV201223} & On-demand & Request Miss Rate & Server-side & Yes & No & No & No \\ \hline
		\cite{Lin:2006:BDD:1143549.1143711}\cite{8309412} & On-demand & Latency & Server-side & No & Yes & No & No \\ \hline
		\cite{Liu:2008:DSM:1626536.1626546} & On-demand & Latency & Server-side & Yes & Yes & No & No \\ \hline
		\cite{He2015118} & On-demand & Request Miss Rate & Server-side & Yes & Yes & Yes & Yes \\ \hline
		\textbf{Our work} & \textbf{On-demand} & \textbf{Latency} & \textbf{Server-side} & \textbf{Yes} & \textbf{Yes} & \textbf{Yes} & \textbf{Yes} \\ \hline
	\end{tabular}
\end{table*}

Lu et al.~\cite{6195838,6216357,6477044} firstly defined two types of well-known data conflicts in multi-channel broadcast applications. They proved the client-side retrieval scheduling problem is $\mathcal{NP}$-hard and provided some client-side data retrieval algorithms for helping clients to retrieve data within multiple channels efficiently. Liu et al.~\cite{Liu:2008:DSM:1626536.1626546} firstly proposed a server-side heuristic data scheduling algorithm, Dynamic Urgency and Productivity (DUP), for on-demand multi-channel systems with consideration of the request conflict (or request overlapping) issue and the dependency between requests for scheduling at the request level and giving higher priorities to the requests which are close to their deadlines. Such an approach provided a counteracting effect to the request starvation problem and improve the utilization of broadcasting bandwidth. However, they did not consider two types of data conflicts. He et al.~\cite{He2015118} proposed a server-side heuristic scheduling approach, most Urgent and Popular request First (UPF), with the consideration of two types of data conflicts in on-demand systems. Except for UPF method, the hardness of data scheduling problem considering two types of data conflicts from the server perspective is seldom discussed. 

The comparisons of the existing works and this paper are summarized in Table~\ref{compared_methods}. In this work, we propose a new server-side heuristic scheduling approach for providing a conflict-free multi-channel data broadcast service with a better performance on the average access time.

\section{Problem Description}
\label{problem}
The length of a broadcasting cycle is an important factor which is normally predefined in the wireless data dissemination applications. Most of existing data scheduling strategies focus on investigating how to efficiently schedule data items in each broadcasting cycle. To validate the performance of a scheduling strategy, average access time (or average latency), is the commonly and widely used metric. If the average access time is shorter, users generally can obtain all the requested data in a shorter time, meaning that the used scheduling strategy is more efficient. In the following subsections, we will describe the considered system model, define the considered scheduling problem, and then prove the hardness of this problem.

\subsection{System Model}
In this work, the considered on-demand multi-channel data dissemination system is shown in Fig.~\ref{fig:system_arch} and we only consider the one-hop broadcasting scenario. The considered data dissemination system uses $|C|+2$ antennas with \emph{Orthogonal Frequency Division Multiplexing} (OFDM) technique~\cite{10.5555/555664} to provide $|C|$ downlink broadcast channels, $1$ downlink index channel and $1$ uplink request channel, where $C=\{c_1,c_2,\dots,c_{|C|}\}$ and $|C|>1$. The downlink index channel and request uplink channel are denoted as $c_{\text{index}}$ and $c_{\text{uplink}}$, respectively. Each user device has two antennas with one for receiving data over the downlink broadcast channels and one for transfering requests via the uplink request channel. We assume that each user device can only access one channel at one time. We assume that all the channels are non-overlapping, synchronous and discretized into fixed-duration slots. The broadcasting server puts the requests coming from the uplink channel into a buffer with \emph{First-Come-First-Serve} (FCFS) strategy and handles all the received requests in a batch manner. In this work, we only focus on the efficiency of (application-layer) data/packet scheduling for users to retrieve the requested data items by accessing the downlink channels.

We assume that all the requested data items are in a dataset $D=\{d_1,d_2,\dots,d_{|D|}\}$, where $|D|$ is the size of $D$, and the length of a broadcasting cycle is $L=|D|$ in default. Suppose that there are $n$ queries, $Q=\{q_1,q_2,\dots,q_n\}$, and each query $q_i$ requests $k$ data items from the dataset $D$, where $i=1,2,\dots,n$ and $k=1,2,\dots,|D|$. We let $q_i=\{d_1^i,d_2^i,\dots,d_k^i\}$ and all the data items have the same data size, where $d_j^i \in D$, $j=1,2,\dots,k$, and $\cup_{i=1}^{n} q_i\subseteq D$. Thus, the system has to arrange the requested data items into $|C|$ broadcasting channels. Note that each time slot on a broadcasting channel can contain at most one data item and data replication is only allowed on different channels. That is, multiple copies of one data item may be placed within a broadcasting cycle. Suppose $L$ is the cycle length, each index $I_t$ at time slot $t$ records the informations about all the data items in time slot $t'$ and the corresponding requests of these data items, where $t'$ is obtained by
\begin{equation}\label{eq1:index_rule}
t' = \left\{
\begin{array}{cl}
(t+2)\mod L, & \mbox{if $t+2> L$} \\
t+2, & \mbox{otherwise.}
\end{array} \right.
\end{equation}
When a client tunes in the channel, it will access the index channel in advance until obtaining information about the first required data item.

\begin{figure}[!t]
	\centering
	\includegraphics[width=0.9\columnwidth]{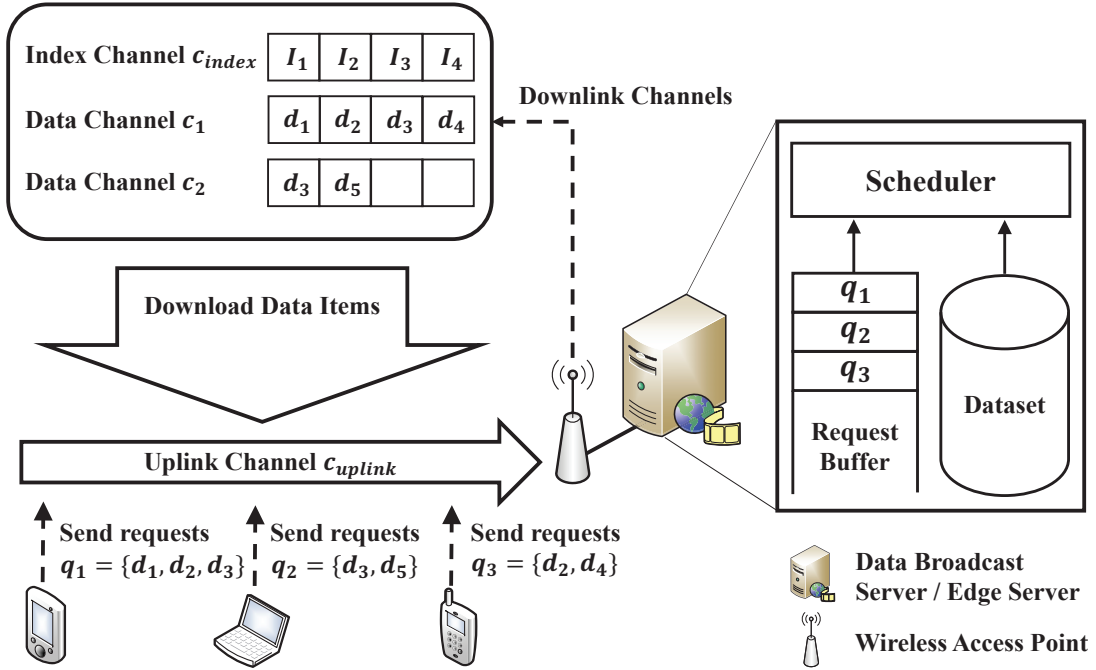}
	\caption{The considered on-demand multi-channel wireless data dissemination environment.}
	\label{fig:system_arch}
\end{figure}

\subsection{Problem Formulation}
The considered scheduling problem can be treated as a mapping $\mathcal{M}$ that data items associated with all the queries to $|C|$ broadcasting channels. For each data item $d_j^i\in D$ associated with a query $q_i\in Q$, let $pos(q_i,d_j^i)=(c_j^i,p_j^i)$ be the position of data item $d_j^i$ in the broadcast, where $c_j^i$ is the channel number, $1\leq c_j^i \leq |C|$ , and $p_j^i$ is the location of $d_j^i$ on that channel, $1\leq p_j^i \leq |D|$. Such a mapping $\mathcal{M}: Q\times D\rightarrow \mathbb{N}\times \mathbb{N}$ is a 1-to-1 mapping.

Since there are multiple channels and each user can only tunes into one broadcasting channel at one time instance, each user may switch channels many times for retrieving all the requested data items on different channels. 
In general, channel switching is a relatively fast operation (in the microseconds range)~\cite{10.1023/B:WINE.0000013082.03518.2e,10.1016/j.pmcj.2005.11.001}. For simplicity, we follow the similar assumptions about channel switching in~\cite{6477044}, and each channel switching takes one time slot in the considered data dissemination environment.
Fig.~\ref{fig:channel_switching} shows an example of the channel switching. However, channel switching may cause a new problem, data conflict, in multi-channel wireless data dissemination systems.
For example, if one of requested data items for request $q_i$ is placed at the previous, the same, or the later location of a scheduled data item which is also associated with $q_i$ on different channels, a data conflict occurs. An example of data conflicts is presented in Fig.~\ref{fig:data_conflict}. The data conflict may result in a longer access time and can be defined as Definition~\ref{data_conflict}. 
\begin{figure}[!t]
	\centering
	\includegraphics[width=0.8\columnwidth]{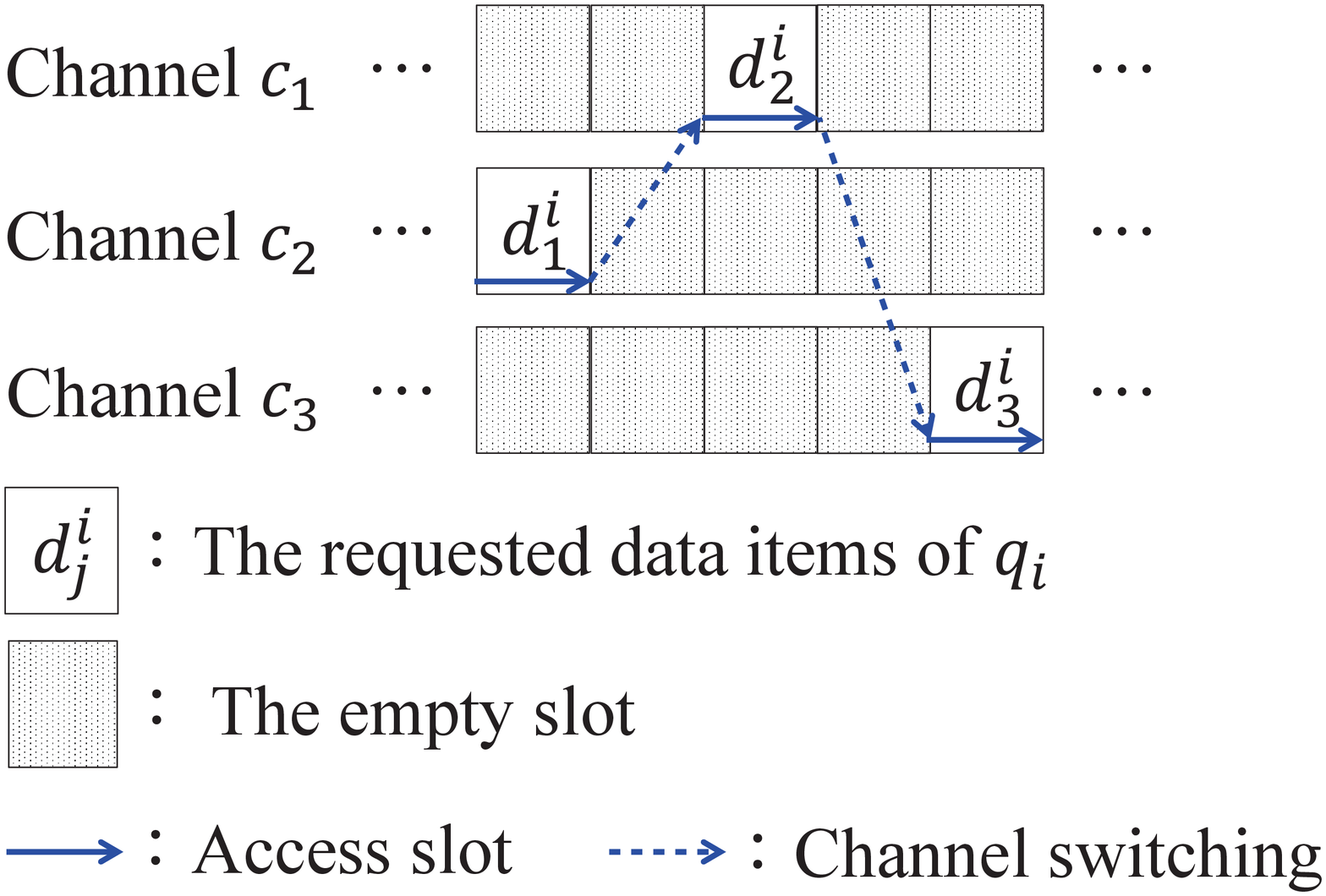}
	\caption{An example of channel switching, where the data items $d_1^i$, $d_2^i$, and $d_3^i$ are requested by $q_i$.}
	\label{fig:channel_switching}
\end{figure}

\begin{definition}[Data Conflict]
	\label{data_conflict}
	For a query $q_i$, two requested data items $d_j^i$ and $d_{j'}^i$, $1\leq j\neq j' \leq k$, if $c_j^i \neq c_{j'}^i$, the conflict occurs when $p_j^i = p_{j'}^i$ or $|p_j^i - p_{j'}^i|=1$.
\end{definition}
\begin{figure}[!t]
	\centering
	\includegraphics[width=0.8\columnwidth]{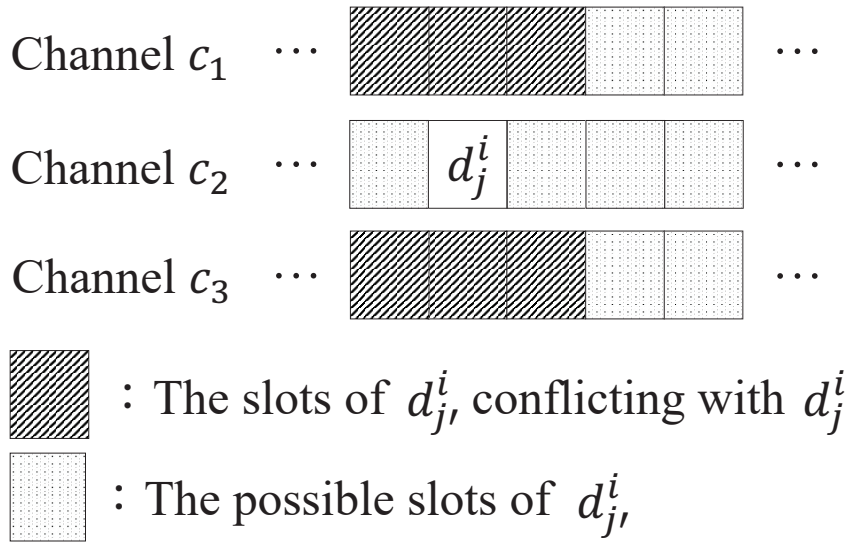}
	\caption{An example of the data conflict problem.}
	\label{fig:data_conflict}
\end{figure}

Let $loc_{min}^i$ denote the minimum value of all the locations of the data items associated with $q_i$ and $loc_{Max}^i$ is the maximum value of all the positions of the data items associated with $q_i$. In other words, $loc_{min}^i=\displaystyle\min_{1\leq j \leq k}p_j^i$ and $loc_{Max}^i=\displaystyle\max_{1\leq j \leq k}p_j^i$. The access time of query $q_i$, $acc(q_i)$, can be defined as $|loc_{Max}^i-loc_{min}^i|$, while the search starts from the beginning of the broadcasting cycle. The average access time for a mapping $\mathcal{M}$ is thus $acc_{\mathcal{M}}=\dfrac{\sum_{i=1}^{n}acc(q_i)}{n}$.
%
%

In summary, the problem we want to solve in this work is \emph{Data Broadcast with Conflict Avoidance} (DBCA) problem which can be defined as follows.
\begin{definition}[DBCA problem]
	\label{conflict_avoid_problem}
	Suppose all the notations are defined as above. The DBCA problem is to find a mapping $\mathcal{M}: Q\times D\rightarrow \{1,\dots,|C|\}\times \{1,\dots,L\}$ such that
	\begin{enumerate}
		\item there is no data conflict for each query in the mapping, i.e., w.r.t. query $q_i$, for each pair of data items $d_j^i$ and $d_{j'}^i$, $1\leq j\neq j' \leq k$, we have  $|p^i_j - p^i_{j'}|>1$ when $c^i_j \not= c^i_{j'}$; and
		\item the average access time of $\mathcal{M}$, $acc_{\mathcal{M}}=\dfrac{\sum_{i=1}^{n}acc(q_i)}{n}$, is minimized.
	\end{enumerate}	
\end{definition}

\subsection{NP-completeness}
To the best of our knowledge, most of the existing works only considered the schedules without data replication in a broadcasting cycle. They did not discuss and analyze the schedules with conflict avoidance problem on multi-channel dissemination environments in detail. Conversely, our proposed approach, FPBS, considers a multi-channel dissemination environment which allows replicating data items on different channels of a broadcasting cycle. In such a scenario, we investigate the data conflict problem and propose a new approach to avoid this problem. In this subsection, we will prove DBCA problem is $\mathcal{NP}$-complete.

In the definition of DBCA problem, the first objective indicates that the broadcasting schedule avoids the data conflict problem. The second objective is to minimize the average access time. Since the server has no prior knowledge about the coming requests, the process for scheduling the broadcasting is made in an online fashion. We first look at the offline version of the DBCA problem in the following and it refers to Conflict-free Data Broadcasting with Minimum average Latency (CDBML) Problem, and define it as below.
\begin{definition}[CDBML problem]
	\label{cfdbml}
	\textbf{Instance:} There are $|C|$ data broadcasting channels with cycle length $L$, a set of $|D|$ data items $D=\{d_1,\dots,d_{|D|}\}$, and a set of $n$ requests $Q=\{q_1, \dots,q_n\}$. Each request $q_i$, $1\leq i\leq n$, is associated with $k$ data items, $d_1^i,d_2^i,\dots,d_k^i$, where $d_j^i\in D$, $1\leq j\leq k\leq |D|$. Any two data items associated with two different requests are different, and every data item needs an unit time $u_t$ to be broadcast. Let $loc_{min}^i$ and $loc_{Max}^i$ be the start time and finish time of $q_i$, respectively.
	
	\textbf{Question:} Does there exist a mapping $\mathcal{M}: Q\times D\rightarrow \{1,\dots,|C|\}\times \{1,\dots,L\}$ such that
	\begin{enumerate}
		\item For two data items $d_j^i$ and $d_j^{i'}$ associated with $q_i$, $|p_j^i-p_{j'}^i|> 1$; and
		\item the average access time, $\sum_{i=1}^{n}|loc_{Max}^i-loc_{min}^i|/n$, is minimized.
	\end{enumerate}	
\end{definition}

In the definition of CDBML problem, the first objective indicates that the broadcasting schedule avoids the data conflict problem. The second objective is to reduce the average access time and all of the data items associated some request $q_i$ should be broadcasting before the end of the broadcasting cycle. $W_i$ is an indication function used to present if a request is served or not. To show further
that the CDBML Problem is $\mathcal{NP}$-complete, we consider a special case of it, where the number of data items associated with each request is the same and equal to the number of channels. That is, we consider the case $k=|C|$. The data items associated with different requests are all different. The following gives the definition of the decision problem for the above special case.
\begin{definition}[CDBML$\rho$ problem]
	\label{dbca_i}
	\textbf{Instance:} There are $|C|$ data broadcasting channels with cycle length $L$, a set of $|D|$ data items $D=\{d_1,\dots,d_{|D|}\}$, a set of $n$ requests $Q=\{q_1, \dots,q_n\}$, and an integer $h$. Each request $q_i$, $1\leq i\leq n$, is associated with $|C|$ data items, $d_1^i,d_2^i,\dots,d_{|C|}^i$, where $d_j^i\in D$, $1\leq j\leq |C|\leq |D|$. Any two data items associated with two different requests are different, and every data item needs an unit time $u_t$ to be broadcast. Let $loc_{min}^i$ and $loc_{Max}^i$ be the start time and finish time of $q_i$, respectively.
	
	\textbf{Question:} Does there exist a mapping $\mathcal{M}: Q\times D\rightarrow \{1,\dots,|C|\}\times \{1,\dots,L\}$ such that
	\begin{enumerate}
		\item For two data items $d_i^j$ and $d_i^{j'}$ associated with $q_i$, $|p_j^i-p_{j'}^i|> 1$; and
		\item $\sum_{i=1}^{n}acc(q_i)/n \leq h$, where $acc(q_i)=|loc_{Max}^i-loc_{min}^i|$.
	\end{enumerate}	
\end{definition}

To show that the CDBML$\rho$ problem is NP-complete, we reduce the \emph{Minimizing Mean flow time in Unit time Open Shop (MMUOS)} scheduling~\cite{Gonzalez1982} problem with preemption ($O|p_{i,j}\in\{0,1\};pmtn|\Sigma C_{i}$) to the CDBML$\rho$ problem. \cite{Gonzalez1982} has proved such a problem ($O|p_{i,j}\in\{0,1\};pmtn|\Sigma C_{i}$) is $\mathcal{NP}$-hard by the reduction from the \emph{graph coloring} problem, and thus the CDBML problem is $\mathcal{NP}$-hard. The MMUOS problem is defined as follows.
\begin{definition}[MMUOS problem]
	\label{mmuos}
	\textbf{Instance:} Given $m$ machines, a set of $n$ jobs $J=\{J_1,J_2,\dots,J_n\}$, a set of $|O|$ unit operations $O=\{o_1,o_2,\dots,o_{|O|}\}$, and an integer $T$. Each job $J_i$, $1\leq i\leq n$, consists of $m$ unit operations $o_i^j$, where $o_i^1,o_i^2,\dots,o_i^m$. The $j_{th}$ operation, $1 \leq j \leq m$, has to be processed on the $j_{th}$ machine. Job $J_i$ will be processed in a window defined by a release time $r_i$ and a finish time $c_i$.
	
	\textbf{Question:} Does there exist a mapping $\mathcal{M}: Q\times D\rightarrow \{1,\dots,|C|\}\times \{1,\dots,L\}$ such that
	\begin{enumerate}
		\item For operations $o_i^j$ and $o_i^{j'}$ in job $J_i$, $\mathcal{M}(o_i^j)\neq \mathcal{M}(o_i^{j'})$; and
		\item $\sum_{i=1}^{n}C_i/n \leq T$, where $C_i=|c_i-r_i|$.
	\end{enumerate}	
\end{definition}

\begin{thm}[]
	\label{optimal_dbca_i}
	The CDBML$\rho$ problem is $\mathcal{NP}$-complete.
\end{thm}
\begin{proof}
	It is easy to see that the CDBML$\rho$ problem is in $\mathcal{NP}$, since validating the existence of an given conflict-free schedule simply needs polynomial time. In order to prove the CDBML$\rho$ problem is $\mathcal{NP}$-hard, a reduction from the MMUOS problem can be made. Suppose that $I'$ is an instance
	of the MNUOS problem. A corresponding instance $I$ of the CDBML$\rho$ problem can be constructed from $I'$ as follows.
	\begin{enumerate}
		\item An unit operation time is equal to the unit time slot to broadcasting a data item.
		\item Let a job $J_i$ correspond to a request $q_i$, $1\leq i\leq m$ and operations $o_i^j$ in $J_i$ be the data item $d_i^j$ associated with $q_i$.
		\item Let $m$ machines be the $|C|$ data broadcasting channels (i.e., $m=|C|$).		
		\item Let $J_i$'s release time $r_i$ be $q_i$'s start time $loc_{min}^i$ in the schedule.
		\item Let $J_i$'s finish time $c_i$ be $q_i$'s finish time $loc_{Max}^i$ in the schedule.
		\item Let integer $T$ be the integer $h$ in CDBML$\rho$ problem.
		\item Let the unit time $u_t'$ in MMUOS problem be three times of $u_t$ in CDBML$\rho$ problem ($u_t'=3*u_t$).
	\end{enumerate}
	According to the last step of the construction, the first objective of MMUOS problem can be equivalent to the first objective of CDBML$\rho$ problem and the above construction can be done in polynomial time. It is straightforward to show that there is a solution for an instance $I'$ of the MMUOS problem if and only if there is a solution for the instance $I$ of the CDBML$\rho$ problem since the reduction is a one-to-one mapping for the variables from the MMUOS problem to the CDBML$\rho$ problem. Hence, the CDBML$\rho$ problem is $\mathcal{NP}$-complete.
\end{proof}

Thus, we can conclude the following theorem.
\begin{thm}[]
	\label{optimal_dbca}
	The CDBML problem is $\mathcal{NP}$-complete.
\end{thm}
%

\section{Frequent Pattern based Broadcast Scheduling}
\label{FPBS}
In this section, we propose an approach, the \emph{Frequent Pattern based Broadcast Scheduling} (FPBS), to shorten the average access time per user for the DBCA problem. In FPBS, we construct a new tree with the frequent patterns of queries. This tree is named as FP*-tree. FPBS includes four stages: (\uppercase\expandafter{\romannumeral1}) sorting requested data items, (\uppercase\expandafter{\romannumeral2}) constructing the FP*-tree's backbone, (\uppercase\expandafter{\romannumeral3}) constructing the FP*-tree's accelerating branches, and (\uppercase\expandafter{\romannumeral4}) schedule mapping. In the following, the proposed method will be introduced with a running example in detail.

\subsection{Stage \uppercase\expandafter{\romannumeral1}: Sorting Requested Data Items}
We consider a running example which uses two data broadcasting channels $c_1$, $c_2$ and an additional index channel $c_{\text{index}}$. The data dissemination server receives five queries $q_1=\{d_2,d_5,d_7\}$, $q_2=\{d_2,d_3,d_4\}$, $q_3=\{d_1,d_3,d_6\}$, $q_4=\{d_1,d_3,d_4,d_5\}$, $q_5=\{d_2,d_5,d_8\}$ and then derives the access frequency $f_{d_j}$ of each data item $d_j$ in these queries. After that, the server sorts all the data items in each query according to the descending order of their access frequencies and also derives the statistical average access frequency $f_{q_i}$ of each query $q_i$
For example, $f_{q_1}=(f_{d_2}+f_{d_5}+f_{d_7})/|q_1|=(3+3+1)/3=2.33$. Hence, the final result is presented in Table~\ref{sorting_data_items}.

The detailed process, $\FuncName{FPBS\_StatisticAndSort}(Q)$, for the first stage is presented in Algorithm~\ref{alg:fpbs:sort}. Line~\ref{alg:fpbs:sort:line1} and Line~\ref{alg:fpbs:sort:line2} analyze the received query set $Q$, derive the statistical information, and save it as a temporary set $S$. The operations from Line~\ref{alg:fpbs:sort:sortRequiredDataByFrequence:start} to Line~\ref{alg:fpbs:sort:sortRequiredDataByFrequence:end} sort every requested data item of each query according to the access frequency of the data item. As the example shown in Table~\ref{sorting_data_items}, the orders of requested data items in queries $q_4$ and $q_5$ change after the sorting. Line~\ref{alg:fpbs:sort:line6} and Line~\ref{alg:fpbs:sort:line7} respectively store the results in two lists, $list_{SortedWithSize}$ and $list_{SortedWithFre}$, in different orders. Finally, the process returns these two lists at Line~\ref{alg:fpbs:sort:end} for the use in following stages.

\begin{algorithm2e}[!t]
	\footnotesize
	\SetAlgoLined
	\Fn{$\FuncName{FPBS\_StatisticAndSort}(Q)$}{
        \KwIn{a set of queries (clients) $Q$}
        \KwOut{two lists of sorted queries with sorted requested data, $list_{SortedWithSize}$, $list_{SortedWithFre}$}
	    create a temporary set $S\leftarrow \phi$\label{alg:fpbs:sort:line1}\;
	    $S\leftarrow$ StatisticDataFrequency($Q$)\label{alg:fpbs:sort:line2}\;
	    \ForEach{query $q$ in $Q$\label{alg:fpbs:sort:sortRequiredDataByFrequence:start}}{
	        sortRequiredDataByFrequency($q$, $S$)\;
	    }\label{alg:fpbs:sort:sortRequiredDataByFrequence:end}
	    $list_{SortedWithSize}\leftarrow$ sortQuerySetByQuerySize($S$)\label{alg:fpbs:sort:line6}\;
	    $list_{SortedWithFre}\leftarrow$ sortQuerySetByAverageFrequency($S$)\label{alg:fpbs:sort:line7}\;
	    \Return $list_{SortedWithSize}$, $list_{SortedWithFre}$\label{alg:fpbs:sort:end}\;
	}
	\caption{Deriving the statistical information and sorted result}
	\label{alg:fpbs:sort}
\end{algorithm2e}

\begin{table}[!ht]
	\caption{The Sorted Result of Requested Data Items}
	\label{sorting_data_items}
	\centering
	\begin{tabular}{llll}
		\hline
		\textbf{Query} & \textbf{Requested data items} & \textbf{Sorted result} & $\bm{f_{q_i}}$\\
		\hline
		$q_1$ & $d_2$, $d_5$, $d_7$ & $d_2$, $d_5$, $d_7$ & 2.33\\
		$q_2$ & $d_2$, $d_3$, $d_4$ & $d_2$, $d_3$, $d_4$ & 2.67\\
		$q_3$ & $d_2$, $d_5$, $d_8$ & $d_2$, $d_5$, $d_8$ & 2.33\\
		$q_4$ & $d_1$, $d_3$, $d_4$, $d_5$ & $d_3$, $d_5$, $d_1$, $d_4$ & 2.5\\
		$q_5$ & $d_1$, $d_3$, $d_6$ & $d_3$, $d_1$, $d_6$ & 2\\
		\hline
	\end{tabular}
\end{table}

\subsection{Stage \uppercase\expandafter{\romannumeral2}: Constructing the FP*-tree's Backbone}
After deriving some statistical information and the sorting result in Table~\ref{sorting_data_items}, the system starts to create the backbone of a FP*-tree. In this stage, the system will always select the query which requests the most number of data items to be inserted into the FP*-tree in advance. If there are multiple queries which request the same number of data items, the system will select the one which has the maximum average access frequency $f_{q_i}$. Thus, the system select $q_4$ as the first query to construct the backbone of a FP*-tree and the result is shown in Fig.~\ref{fig:fpbs}\subref{fig:fpbs:q4}. After adding $q_4$ to the FP*-tree, the system will update the statistical information of unhandled queries, as shown in Table~\ref{info_after_add_q4}.

After updating the statistical information, the system will select the next query to handle in the same way. In the previously mentioned, both $q_1$ and $q_3$ request 2 data items so the system will compare the remaining average access frequencies of $q_1$ and $q_3$ ($f_{q_1}=2$, $f_{q_3}=2$) and both values are the same. Then the query which comes into the system first will be selected, so $q_1$ becomes the next one in this step. Note that the numbering of $q_1$ is smaller than $q_3$'s and it means that $q_1$ comes into the system earlier. Thus the handling priority of the remaining queries is $q_1\rightarrow q_3\rightarrow q_5$. While adding data item $d_2$ into the FP*-tree, the system needs to consider the relations between $d_2$ and the other queries. In this case, $q_2$ and $q_3$ also request the data item $d_2$. The system then checks the other data items which are in the request list of both queries and have been added into the FP*-tree. Since the level of $d_4$ is larger than $d_3$'s level, the system will insert $d_2$ as $d_4$'s child. Such a way can avoid increasing the access time of $q_4$ which has been handled. After handling $d_2$, the system handles $d_7$ in the same way and the result of FP*-tree is shown in Fig.~\ref{fig:fpbs}\subref{fig:fpbs:q1}. The system then updates the statistical information which is presented in Table~\ref{info_after_add_q1}.

\begin{figure}[!t]
	\centering
	\subfigure[Add $q_4$]{
		\label{fig:fpbs:q4} 
		\includegraphics[width=0.103 \columnwidth]{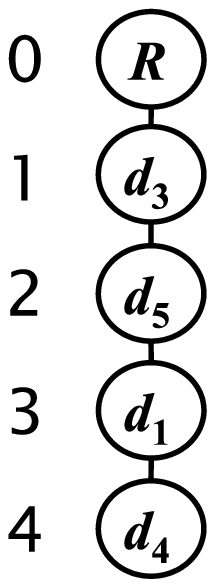}}\hspace{0.035\columnwidth}
	\subfigure[Add $q_1$]{
		\label{fig:fpbs:q1} 
		\includegraphics[width=0.105 \columnwidth]{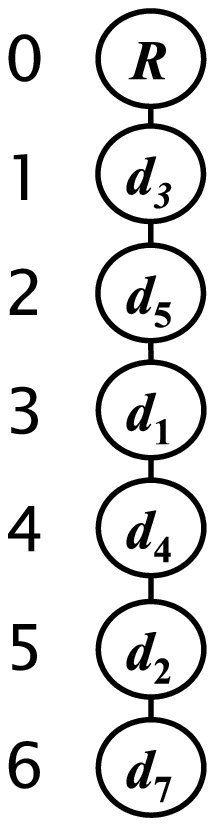}}\hspace{0.025\columnwidth}
	\subfigure[Add $q_3$]{
		\label{fig:fpbs:q3} 
		\includegraphics[width=0.144 \columnwidth]{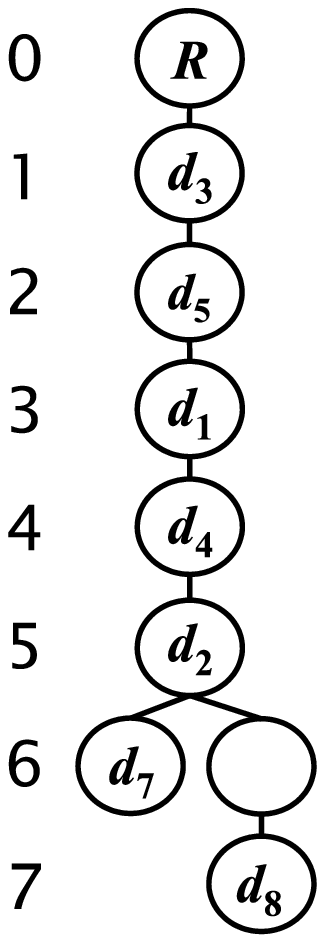}}\hspace{0.005\columnwidth}
	\subfigure[Add $q_5$]{
		\label{fig:fpbs:q5} 
		\includegraphics[width=0.186 \columnwidth]{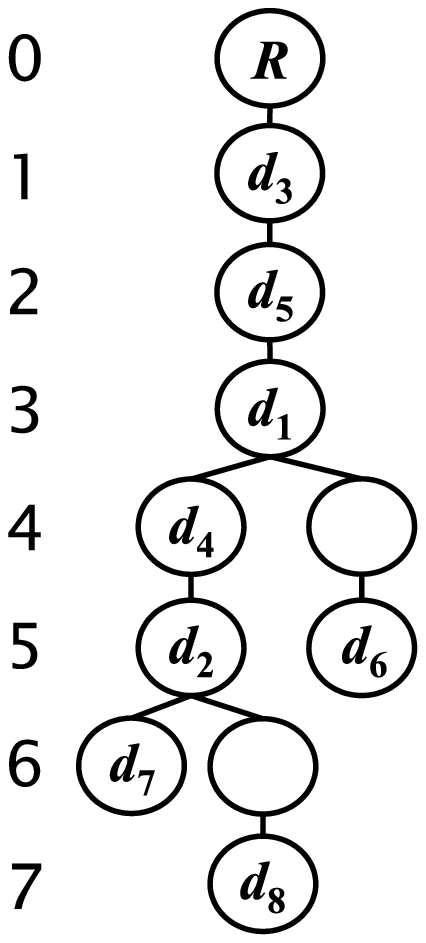}}
	\caption{Constructing the backbone of a FP*-tree step-by-step: \subref{fig:fpbs:q4} add $q_4$, \subref{fig:fpbs:q1} add $q_1$, \subref{fig:fpbs:q3} add $q_3$, and \subref{fig:fpbs:q5} add $q_5$.}
	\label{fig:fpbs} 
\end{figure}

Next query which will be handled is $q_3$. Since there are no other queries relating to the requested data item $d_8$, the system needs to add $d_8$ after $d_2$ according to the order of $q_3$'s requested list. However, $d_2$ is also requested by $q_1$ and thus $d_2$ already has one branch and the position is occupied by $d_7$. Therefore, $d_8$ can only be scheduled in the level (time slot) after $d_2$ and $d_7$. In this case, the system creates a new branch of $d_2$ and inserts an empty node between $d_2$ and $d_8$. Note that an empty node is a node without saving any data item. After handling $q_3$, the results are shown in Fig.~\ref{fig:fpbs}\subref{fig:fpbs:q3}. The last query is $q_5$ and there are no other queries relating to $d_6$. Hence, the system has to add $d_6$ after $d_1$ according to the order of $q_5$'s requested list. However, $d_1$ is also requested by $q_4$ and thus $d_6$ needs to be scheduled after $d_4$. In this case, the system creates a new branch of $d_1$ and inserts an empty node between $d_1$ and $d_6$. Finally, the construction of FP*-tree's backbone is finished and the result is shown in Fig.~\ref{fig:fpbs}\subref{fig:fpbs:q5}.

\begin{table}[!t]
	\caption{Updated Result After Handling $q_4$}
	\label{info_after_add_q4}
	\centering
	\begin{tabular}{lp{0.2\columnwidth}p{0.35\columnwidth}l}
		\hline
		\textbf{Query} & \textbf{Unhandled data items} & \textbf{Items added in FP*-tree's backbone} & $\bm{f_{q_i}}$\\
		\hline
		$q_1$ & $d_2$, $d_7$ & $d_5$ & 2\\
		$q_2$ & $d_2$ & $d_3$, $d_4$ & 3\\
		$q_3$ & $d_2$, $d_8$ & $d_5$ & 2\\
		$q_4$ & $\emptyset$ & $d_3$, $d_5$, $d_1$, $d_4$ & 0\\
		$q_5$ & $d_6$ & $d_3$, $d_1$ & 1\\
		\hline
	\end{tabular}
\end{table}
\begin{table}[!t]
	\caption{Updated Result After Handling $q_1$}
	\label{info_after_add_q1}
	\centering
	\begin{tabular}{lp{0.2\columnwidth}p{0.35\columnwidth}l}
		\hline
		\textbf{Query} & \textbf{Unhandled data items} & \textbf{Items added in FP*-tree's backbone} & $\bm{f_{q_i}}$\\
		\hline
		$q_1$ & $\emptyset$ & $d_5$, $d_2$, $d_7$ & 0\\
		$q_2$ & $\emptyset$ & $d_3$, $d_4$, $d_2$ & 0\\
		$q_3$ & $d_8$ & $d_5$, $d_2$ & 1\\
		$q_4$ & $\emptyset$ & $d_3$, $d_5$, $d_1$, $d_4$ & 0\\
		$q_5$ & $d_6$ & $d_3$, $d_1$ & 1\\
		\hline
	\end{tabular}
\end{table}

Algorithm~\ref{alg:fpbs:backbone} presents two functions for the backbone construction. $\FuncName{FPBS\_CreateBackbone}(S)$ describes the main process of an FP*-tree's backbone construction and $\FuncName{FPBS\_AddNodeForBackbone}(\mathcal{T},N_{p},d)$ is the function of adding a node during the backbone construction. From Line~\ref{alg:fpbs:backbone:line2} to Line~\ref{alg:fpbs:backbone:line5}, the operations initialize an empty FP*-tree $\mathcal{T}$ and create a sorted query table $Q_{table}$ with the derived sorted result in the stage \uppercase\expandafter{\romannumeral1}. The operations from Line~\ref{alg:fpbs:backbone:firstq:start} to Line~\ref{alg:fpbs:backbone:firstq:end} handle each requested data item of the first query in the sorted query set. The first query is the most important and has maximum number of requested data items. As shown as the above example in Fig.~\ref{fig:fpbs}\subref{fig:fpbs:q4}, the query $q_4$ is the first to be handled. At Line~\ref{alg:fpbs:backbone:line9}, the remaining information of unhandled queries and data items in the query table $Q_{table}$ will be updated. From Line~\ref{alg:fpbs:backbone:line10} to Line~\ref{alg:fpbs:backbone:line17}, the operations continuous inserting the unhandled data items of $Q_{table}$ into the backbone of $\mathcal{T}$. At Line~\ref{alg:fpbs:backbone:line13}, the operation finds the right position of $\mathcal{T}$'s backbone to insert the unhandled data item with the consideration of query dependency and the order of data items. The operations from Line~\ref{alg:fpbs:backbone:addnode:start} to Line~\ref{alg:fpbs:backbone:addnode:end} presents the detailed process of adding a data node to the backbone of $\mathcal{T}$. Note that the operation, $\mathcal{T}.isOverload(N_{temp}.slot+1)$, at Line~\ref{alg:fpbs:backbone:overload} is used to avoid scheduling data items out of $|C|$ data broadcasting channels. Fig.~\ref{fig:fpbs:q3} and Fig.~\ref{fig:fpbs:q5} are the running examples for such operations.
\begin{algorithm2e}[!t]
	\scriptsize
	\SetAlgoLined
	\Fn{$\FuncName{FPBS\_CreateBackbone}(S)$}{
	    \KwIn{a sorted set of queries (clients) $S$}
	    \KwOut{a basic FP*-tree $\mathcal{T}$}
	    create a empty FP*-tree $\mathcal{T}$ and the root $R$ of $\mathcal{T}$\label{alg:fpbs:backbone:line2}\;
	    set $S$ into a query table $Q_{table}$\;
	    let $q\leftarrow$ $S$.first()\;
	    let a temporary pointer $N_{curr}\leftarrow R$\label{alg:fpbs:backbone:line5}\;
	    \ForEach{requested data $d$ in $q$\label{alg:fpbs:backbone:firstq:start}}{
		    $N_{curr}\leftarrow$ $\FuncName{FPBS\_AddNodeForBackbone}(\mathcal{T},N_{curr},d)$\;
	    }\label{alg:fpbs:backbone:firstq:end}
	    update $Q_{table}$\label{alg:fpbs:backbone:line9}\;
	    \While{$Q_{table}$ contains any unhandled required data}{
	    \label{alg:fpbs:backbone:line10}
	        $q_{un}\leftarrow$ the query with the maximum number of unhandled data items in $Q_{table}$\;
	        \ForEach{unhandled requested data $d'$ in $q_{un}$}{
		        $N_{d\_p}\leftarrow$ find the other queries which also needs data $d'$ and then choose one of the handled data nodes whose slot is maximum in $\mathcal{T}$\label{alg:fpbs:backbone:line13}\;
                $\FuncName{FPBS\_AddNodeForBackbone}(\mathcal{T},N_{d'\_p},d')$;
            }
	        update $Q_{table}$\;
	    }\label{alg:fpbs:backbone:line17}
	    \Return $\mathcal{T}$\label{alg:fpbs:backbone:createbackbone:end}\;
	}
	
	\Fn{$\FuncName{FPBS\_AddNodeForBackbone}(\mathcal{T},N_{p},d)$}{
	    \KwIn{an FP*-tree $\mathcal{T}$, the parent node $N_{p}$, and a new data item $d$}
	    \KwOut{an added node $N_d$}
        create a new node $N_d$ with data item $d$\label{alg:fpbs:backbone:addnode:start}\;
	    \uIf{$N_{p}$ has children}{
	        create an empty node $N_{e}$\;
	        $N_{p}$.addChild($N_{e}$)\;
	        let a temporary pointer $N_{temp}\leftarrow N_{e}$\label{alg:fpbs:backbone:line25}\;
            \While{$\mathcal{T}$.isOverload($N_{temp}$.slot+1)}{\label{alg:fpbs:backbone:overload}
                create an empty node $N_{e}$\;
	            $N_{p}$.addChild($N_{e}$)\;
	            $N_{temp}\leftarrow N_{e}$\;
            }
            $N_{temp}$.addChild($N_d$)\;
	    }
	    \Else{
            $N_p$.addChild($N_d$)\;
	    }
	    \Return $N_d$\label{alg:fpbs:backbone:addnode:end}\;
	}
	
	\caption{Functions used for the FP*-tree's backbone construction}
	\label{alg:fpbs:backbone}
\end{algorithm2e}

\subsection{Stage \uppercase\expandafter{\romannumeral3}: Constructing the FP*-tree's accelerating branches}
After the construction of FP*-tree's backbone, the system starts to create the accelerating branches to optimize the schedule. The purpose of constructing the accelerating branch is to increase the chance of each user getting the requested data item earlier after switching channels. 

In this stage, we propose two different ordering rules, \emph{Request-Number-First} and \emph{Frequency-First}, to insert data items in the FP*-tree's accelerating branches. The priority of a query for the insertion of FP*-tree is decided by following values: number of requested data items, average access frequency, and arrival time. With Request-Number-First rule, the system will select the query which requests the maximum number of data items to handle first. If multiple queries request the maximum number of data items, the system will select the one of them that has the maximum average access frequency. If multiple queries has the maximum average access frequency unfortunately, the system will select the query according to its arrival order. Conversely, with Frequency-First rule, the system will first select the query which has the maximum average access frequency. If multiple queries has the maximum average access frequency, the system will select the one of them that requests the maximum number of data items. If multiple queries requests the maximum number of data items unfortunately, the system will select the query according to its arrival order. Note that the construction of the FP*-tree's backbone always follows Request-Number-First rule in our design. The system can use different rules only when constructing accelerated branches of the FP*-tree. 

Since different orders of handling queries and data items make the process constructs different accelerating branches of FP*-trees, we will compare the performance results of different schedules generated by using different rules. By default, the system uses Frequency-First rule to select the query for constructing the FP*-tree's accelerating branches. Due to limitations on space and the similar process, we only introduce the proposed approach with Frequency-First in detail. In this example, the system follows Frequency-First rule and gets the following handling sequence, $q_2\rightarrow q_4\rightarrow q_1\rightarrow q_3\rightarrow q_5$. Note that the value of $f_{q_i}$ is shown in Table~\ref{sorting_data_items}.

The system first handles query $q_2$ and $q_2$'s sorted requested data items are $d_2$, $d_3$, and $d_4$. Hence, the system sequentially schedules $d_2$, $d_3$, and $d_4$. When scheduling $d_2$, the system temporarily inserts $d_2'$ into level (or slot) 1 and the position is a right child of the root. Then the system searches $d_2$ in the backbone and check whether $p_2>p_2'$ and $p_2-p_2'>1$ or not. In this case, $p_2>p_2'$ and $p_2-p_2'=4>1$ is hold, so $d_2$ can be inserted into the position of $d_2'$. For the next requested data item $d_3$, the system inserts $d_3'$ after $d_2$ in the accelerating branch and then checks whether the position is legal or not in the same way. In this case, $d_3$ can be inserted into the position of $d_3'$. For the last requested data item $d_4$ by query $q_2$, the system tries to temporarily insert $d_4'$ after $d_3$ in the accelerating branch. However, the system can find $d_4$ in the backbone that $p_4-p_4'\leq 1$. Thus, $d_4$ can not be inserted into the accelerating branch. After handling $q_2$, the result of FP*-tree is shown in Fig.~\ref{fig:fpbs:ac}\subref{fig:fpbs:ac:q2}.

For the next query $q_4$, the system will do nothing in the accelerating branch. The reason is that $q_4$ is the first query handled in the backbone and the schedule, $d_3\rightarrow d_5\rightarrow d_1\rightarrow d_4$, has been optimized. Go on the next step, $q_1$ is going to be handled and $q_1$'s requested data items are $d_2$, $d_5$, and $d_7$. Since $d_2$ has been inserted into the accelerating branch, the system skips $d_2$ and tries to insert $d_5$ in this step. According to the order of $q_1$'s requested list, $d_5$ needs to be inserted after $d_2$. In the accelerating branch, node $d_2$ already has a child, so the system creates a new branch of $d_2$, inserts an empty node as $d_2$'s right child, and then add temporary $d_5'$ after the empty node. Since there is no $d_5$ whose $p_5>p_5'$ in the backbone, it is legal to insert $d_5$ at the position of $d_5'$. For the last requested data item $d_7$ in $q_1$, $d_7$ is inserted in the same way. The system inserts $d_7'$ after $d_5$ in advance and check whether the backbone contains $d_7$ or not. Sine $p_7>p_7'$ and $p_7-p_7'=2>1$, it is legal to insert $d_7$ at the position of $d_7'$. After handling all the requested data items in $q_1$, the result of FP*-tree is shown in Fig.~\ref{fig:fpbs:ac}\subref{fig:fpbs:ac:q1}.

\begin{figure}[!t]
	\centering
	\subfigure[Add $q_2$]{
		\label{fig:fpbs:ac:q2} 
		\includegraphics[width=0.2417 \columnwidth]{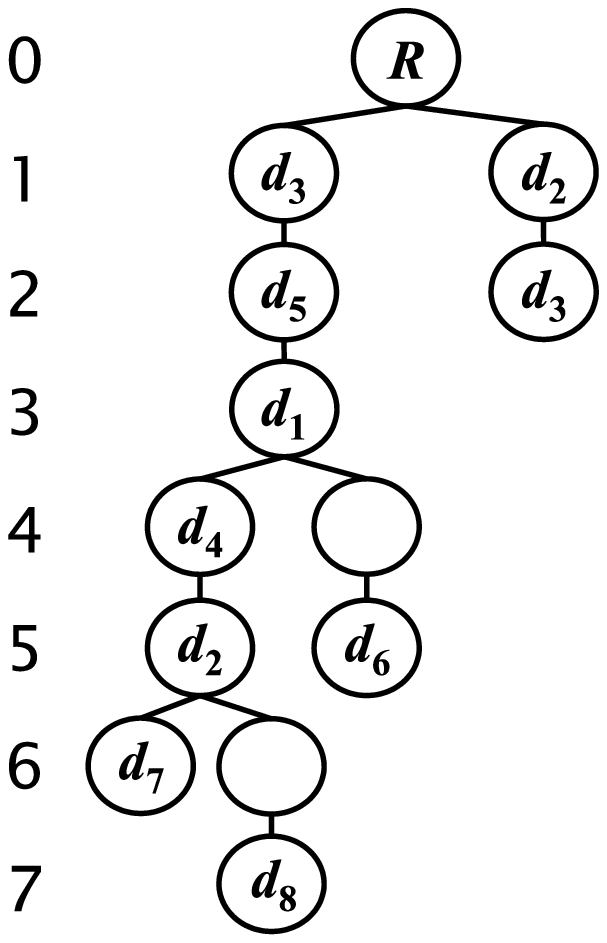}}
	\subfigure[Add $q_1$]{
		\label{fig:fpbs:ac:q1} 
		\includegraphics[width=0.2835 \columnwidth]{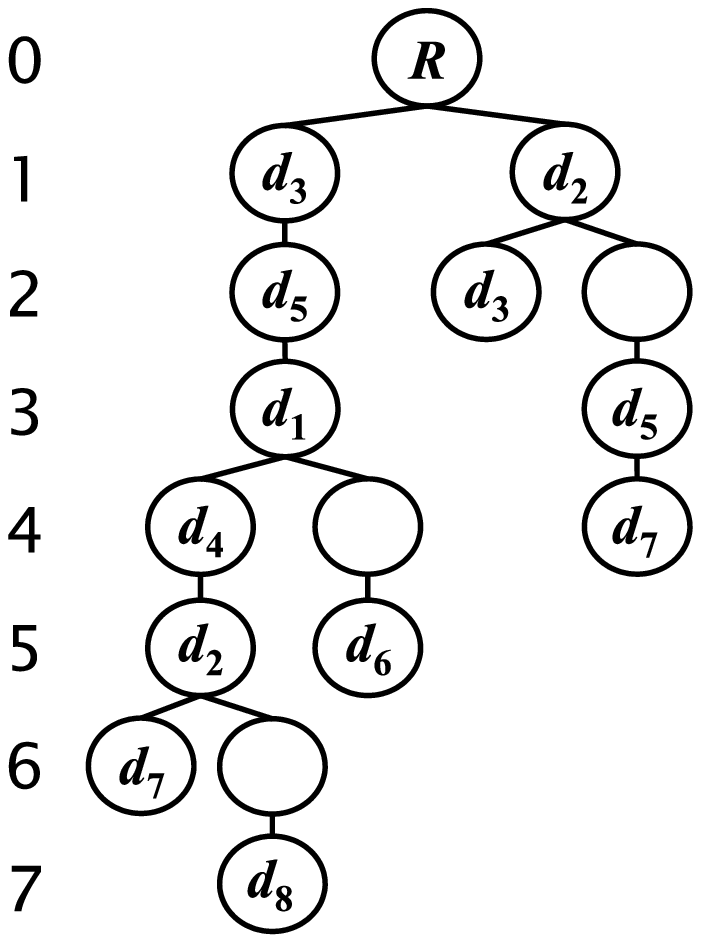}}
	\subfigure[Add $q_5$]{
		\label{fig:fpbs:ac:q5} 
		\includegraphics[width=0.2996 \columnwidth]{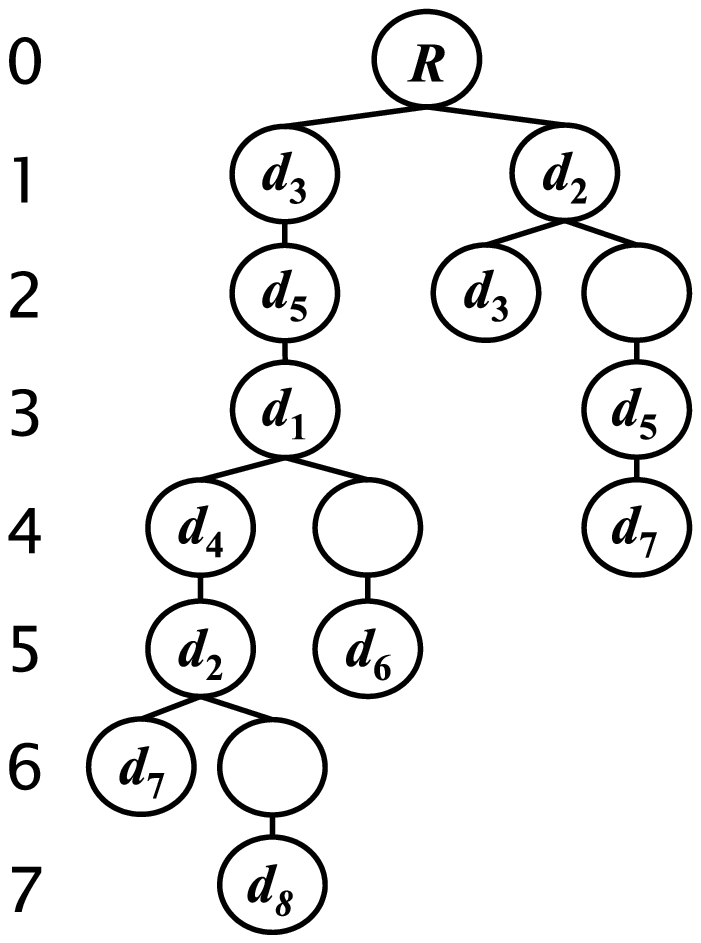}}
	\caption{Constructing the accelerating branch of the FP*-tree step-by-step: \subref{fig:fpbs:ac:q2} add $q_2$, \subref{fig:fpbs:ac:q1} add $q_1$, and \subref{fig:fpbs:ac:q5} add $q_5$.}
	\label{fig:fpbs:ac} 
\end{figure}

After handling $q_1$, the system will start to handle $q_3$. The sorted requested data items are $d_2$, $d_5$, $d_8$. Since $d_2$ has been scheduled at the first slot (level) in the accelerating branch, the system skips $d_2$ in this step. The next data item $d_5$ also has been scheduled in the accelerating branch while handing the previous query $q_1$. Hence, the system only needs to handle $d_8$ for $q_4$. According to the requested list of $q_4$, $d_8$ needs to be inserted at a position that is after $d_2$ and $d_5$. In the accelerating branch, $p_5>p_2$ so that $d_8$ will be inserted under the $d_5$. Since $d_5$ already has a branch, the system creates a new branch of $d_5$, inserts an empty node after $d_5$, and tries to inserts a temporary $d_8'$ after the empty node (at $p_8'=5$). However, $C=2$ and the bandwidth has been occupied by $d_2$ and $d_6$ at slot $p_8'=5$. Then the system will insert an empty node again and try to add a temporary $d_8'$ at position $p_8'=6$. Then the system starts to find $d_8$ in the backbone and check whether $p_8>p_8'$ and $p_8-p_8'>1$ or not. In this case, $p_8-p_8'=1$, so it is illegal to place $d_8$ at the position of $d_8'$ and the system removes all the empty nodes after $d_5$ in the accelerating branch. Hence, the final FP*-tree is shown in Fig.~\ref{fig:fpbs:ac}\subref{fig:fpbs:ac:q5}.
\begin{algorithm2e}[!t]
	\scriptsize
	\SetAlgoLined
	\Fn{$\FuncName{FPBS\_CreateAcceleratingBranch}(\mathcal{T},S)$}{
		\KwIn{an FP*-tree $\mathcal{T}$ and a sorted set of queries (clients) $S$}
		\KwOut{a final FP*-tree $\mathcal{T}$}
		let a temporary pointer $N_{curr}\leftarrow R$\;
		create a temporary list $list_q$ and a temporary node $N_{temp}$\;
		\ForEach{query $q$ in $S$\label{alg:fpbs:acceleratingbranch:firstfor:start}}{
			\ForEach{requested data $d$ in $q$\label{alg:fpbs:acceleratingbranch:secondfor:start}}{
				$N_{temp}\leftarrow \FuncName{FPBS\_AddNodeForAcceleratingBranch}(\mathcal{T},$ $N_{curr}, d)$\label{alg:fpbs:acceleratingbranch:sub1}\;
				$N_{curr}\leftarrow \FuncName{RangeSearch}(\mathcal{T}$, $N_{temp}$)\label{alg:fpbs:acceleratingbranch:sub2}\;
				$list_q$.add($N_{curr}$)\;
				\If{$N_{curr}$.slot $>$ $\mathcal{T}$.slot}{
					delete the path of $list_q$ in $\mathcal{T}$\;
					break\;
				}
			}\label{alg:fpbs:acceleratingbranch:secondfor:end}
			$list_q$.clear()\;
		}\label{alg:fpbs:acceleratingbranch:firstfor:end}
		\Return $\mathcal{T}$\;
	}
	
	\Fn{$\FuncName{FPBS\_AddNodeForAcceleratingBranch}(\mathcal{T},N_{p},d)$}{
		\KwIn{an FP*-tree $\mathcal{T}$, the parent node $N_{p}$, and a new data item $d$}
		\KwOut{an added node $N_d$}
		create a new node $N_d$ with data item $d$\;
		\uIf{$N_{p}$ has children}{
			\uIf{$N_{p}$ has a child $N_d'$ with $d$}{
				\Return $N_d'$;
			}
			\Else{
				create an empty node $N_{e}$\;
				$N_{p}$.addChild($N_{e}$)\;
				let a temporary pointer $N_{temp}\leftarrow N_{e}$\;
				\While{$\mathcal{T}$.isOverload($N_{temp}$.slot+1)}{\label{alg:fpbs:acceleratingbranch:check_overload}
					create an empty node $N_{e}$\;
					$N_{p}$.addChild($N_{e}$)\;
					$N_{temp}\leftarrow N_{e}$\;
				}
				create a new node $N_d$ with data item $d$\;
				$N_{temp}$.addChild($N_d$)\;
			}
		}
		\Else{
			create a new node $N_d$ with data item $d$\;
			$N_p$.addChild($N_d$)\;
		}
		\Return $N_d$\;
	}
	
	\Fn{$\FuncName{FPBS\_RangeSearch}(\mathcal{T},N_{p})$}{
		\KwIn{an FP*-tree $\mathcal{T}$ and a search node $N_d$}
		\KwOut{a result node $N_{d}$ within the search range}	
		\Int{} $Num_e\leftarrow$ the number of $N_d$'s ancestors which are empty\;
		\Int{} $startSlot\leftarrow N_{temp}$.slot $– Num_e + 1$\;
		\Int{} $endSlot\leftarrow N_{temp}$.slot $+ 1$\;
		\For{$i\leftarrow startSlot$ \KwTo $endSlot$\label{alg:fpbs:acceleratingbranch:rangesearch:firstfor:start}}{
			\If{find a node $N_{temp}$ that has the same data as $N_d$ does at level $i$ of $\mathcal{T}$}{
				delete the path that contains $N_d$ and all the empty connected ancestors of $N_d$\label{alg:fpbs:acceleratingbranch:delete_inserted_nodes}\;
				\Return $N_{temp}$\;
			}
		}\label{alg:fpbs:acceleratingbranch:rangesearch:firstfor:end}
		\Return $N_d$\label{alg:fpbs:acceleratingbranch:end}\;
	}
	\caption{Functions used for the construction of the FP*-tree's accelerating branch}
	\label{alg:fpbs:acceleratingbranch}
\end{algorithm2e}

Algorithm~\ref{alg:fpbs:acceleratingbranch} presents the pseudo-codes for the functions of accelerating branch construction. $\FuncName{FPBS\_CreateAcceleratingBranch}(\mathcal{T},S)$ is the main function for constructing accelerating branch. The process calls the sub-function $\FuncName{FPBS\_AddNodeForAcceleratingBranch}(\mathcal{T},N_{curr}, d)$ to insert a data item into the accelerating branch of $\mathcal{T}$ at Line~\ref{alg:fpbs:acceleratingbranch:sub1}. Such a process is similar to the function $\FuncName{FPBS\_AddNodeForBackbone}(\mathcal{T},N_{p},d)$ in the backbone construction. The operation at Line~\ref{alg:fpbs:acceleratingbranch:sub2} calls another sub-function $\FuncName{FPBS\_RangeSearch}(\mathcal{T},N_p)$ to check whether the inserted data item is in the search range (or levels)) or not. The insertion will be illegal if the same data item in the backbone of $\mathcal{T}$ locates at one of search levels. If the insertion is illegal, the inserted nodes (including the data item and empty node(s)) will be deleted at Line~\ref{alg:fpbs:acceleratingbranch:delete_inserted_nodes}.

\subsection{Stage \uppercase\expandafter{\romannumeral4}: Schedule Mapping}
After finishing stage III, the system will map every slot (or level) of FP*-tree into the broadcasting channels using the Breadth-First-Search (BFS) strategy. The final results are shown in Fig.~\ref{fig:fpbs:result}. Note that the maximum number of data items in each slot (level) is the number of channels, $|C|$. The mapping process is described as the operations before Line~\ref{alg:fpbs:mapping:firstwhile:end} in Algorithm~\ref{alg:fpbs:mapping}. From Lines~\ref{alg:fpbs:indexing:start} to~\ref{alg:fpbs:indexing:end}, the process schedules the index items in index channel and the result is shown in Fig.~\ref{fig:fpbs:result}. According to the indexing rule defined in~\eqref{eq1:index_rule}, the index $I_1$ records the information about who requests the data items in slot $3$ and the index $I_6$ records the similar information corresponding to the data items in slot $1$.

Consider the example of the Table I, for the request $q_2=\{d_2,d_3,d_4\}$, the final schedule in Fig.~\ref{fig:fpbs:result} generated by the proposed FPBS shows that the user can retrieve all the requested data items $d_2$, $d_3$ (on $c_2$), and $d_4$ (on $c_1$) within 4 time slots including a channel switching. If there is no accelerating branch, the user needs 5 time slots to retrieve data items $d_2$, $d_3$, and $d_4$ on $c_1$. This result shows that the proposed FP*-tree can indeed reduce the access time.

\begin{algorithm2e}[!t]
	\footnotesize
	\SetAlgoLined
	\Fn{$\FuncName{FPBS\_ScheduleMapping}(\mathcal{T},S,|C|)$}{
	    \KwIn{an FP*-tree $\mathcal{T}$, a sorted query set $S$, and the munber of channels $|C|$}
	    \KwOut{a scheduled channel set $S_{channel}$ and a index channel $I_{channel}$}	
	    let a list $list_{handling}\leftarrow T$.root.children\;
	    let a temporary list $list_{next}\leftarrow \emptyset$\;
	    create a data channel $S_{channel}$ with $|C|$ data broadcasting channels (or rows)\; 
	    create an index channel $I_{channel}\leftarrow \emptyset$\;
	    \Int{} $i$\;
	    \While{$list_{handling}$ is not empty\label{alg:fpbs:mapping:firstwhile:start}}{
	        $i\leftarrow 1$\tcc*[r]{$i$ is used as a pointer to the current channel}
	        \ForEach{node $N$ in $list_{handling}$\label{alg:fpbs:mapping:firstfor:start}}{
	            \uIf{$N$ is an empty node}{
                    break\;
                }
                \uElseIf{$N$.parent is an empty node}{
                    insert $N$ into $S_{channel}$ whose slot $N$.slot is not occupied\;
                }
                \Else{
                    insert $N$ into the $i$th channel\;
                }
                \If{$N$ is not a leaf node}{
                    add $N$'s children into $list_{next}$\;
                }
                $i\leftarrow i+1$\;
            }\label{alg:fpbs:mapping:firstfor:end}
            copy every node of $list_{next}$ to $list_{handling}$\;
            $list_{next}$.clear()\;
	    }\label{alg:fpbs:mapping:firstwhile:end}
    
    	\For{$j\leftarrow 1$ to $\mathcal{T}$.height()\label{alg:fpbs:indexing:start}}{
    		\For{$i\leftarrow 1$ to $|C|$}{
    			Use $S$ to check who requests the data item in the slot determined by~\eqref{eq1:index_rule} and channel $C_i$ of $S_{channel}$ and then update this information to $I_{channel}[j]$\;
    		} 
    	}\label{alg:fpbs:indexing:end}
	    \Return $I_{channel}$, $S_{channel}$\label{alg:fpbs:mapping:end}\;
	}
	\caption{The function used for the schedule mapping}
	\label{alg:fpbs:mapping}
\end{algorithm2e}

\begin{figure}[!ht]
	\centering
	\includegraphics[width=0.9 \columnwidth]{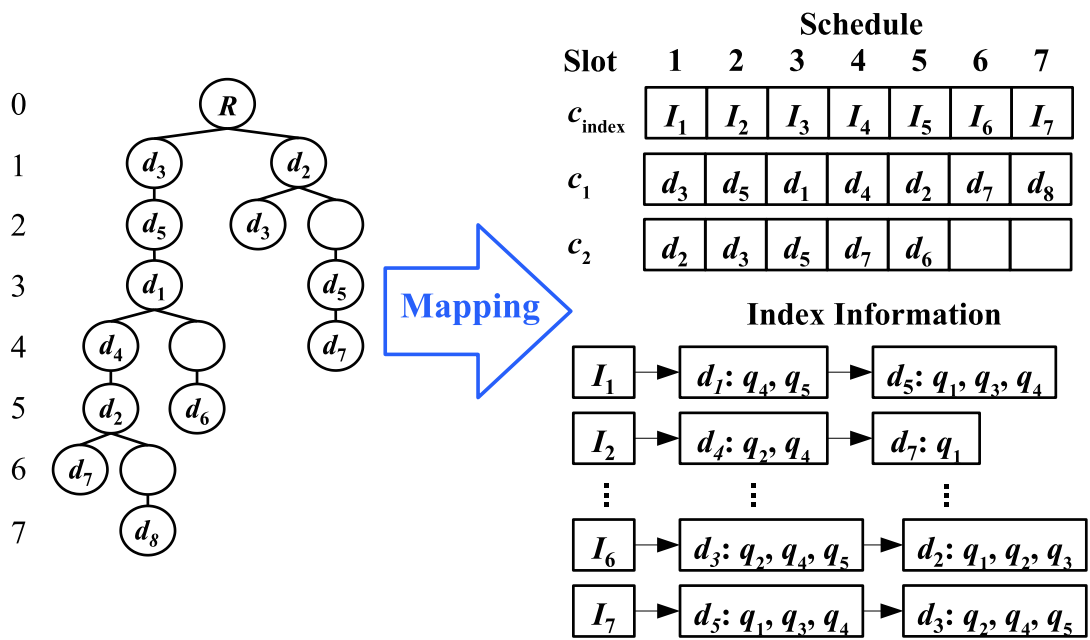}
	\caption{Mapping FP*-tree into the broadcasting channels with indexes.}
	\label{fig:fpbs:result}
\end{figure}


\section{Analysis and Discussion}
\label{analysis}
In this section, we analyze the performance of FPBS in terms of \emph{time complexity}, \emph{space complexity} and \emph{access time}.

\subsection{Time Complexity}
Suppose that the notations are defined as above and the FP*-tree is denoted as $\mathcal{T}$, then the time complexity of the $\mathcal{T}$'s construction will be $\mathcal{O}(nk)$. The idea of FP*-tree design comes up from the FP-tree and only one difference between them is that FP*-tree needs to add an empty node when creating a new branch except for the root node. In the last stage of the proposed method, schedule mapping needs to maps all the data nodes of $\mathcal{T}$ to the broadcasting channels and $|\mathcal{T}|\leq nk$, so the time complexity of schedule mapping is also $\mathcal{O}(nk)$. Due the to nature of the FP*-tree which is evolved from FP-tree, FPBS costs $\mathcal{O}(nk)$ in both average case and worst case. In summary, FPBS provides a polynomial algorithm for solving the DBCA problem. 

\subsection{Space Complexity}
After discuss the time complexity of FPBS, we starts to analyze the space complexity of FPBS. In this part, we only consider the temporary space for FPBS process. In the stage \uppercase\expandafter{\romannumeral1} of FPBS process, the system uses a $\mathcal{O}(nk)$ size table to store the sorted requests and the statistical information. In the stage \uppercase\expandafter{\romannumeral2}, the system uses the obtained sorted table to construct the backbone of an FP*-tree and it also costs $\mathcal{O}(nk)$ space. In the stage \uppercase\expandafter{\romannumeral3}, the system constructs accelerating branches of the FP*-tree and it costs $\mathcal{O}(nk')$ space, where $1\leq k'\leq k$. In the last stage, the system just maps the FP*-tree to the channels and only costs $\mathcal{O}(1)$ additional temporary space for traversing the FP*-tree. That is, the temporary space complexity during the scheduling process is $\mathcal{O}(nk)$. 

\subsection{Access Time}
In wireless data dissemination environments, access time (or latency) is an important metric for validating the efficiency of scheduling. 
In FPBS, the system always first selects the request, whose size and average access frequency are maximum, and then schedules it in the backbone of FP*-tree. We then treat is as the base of schedule. That is, the access time for a request $q_i$ can be formulated as Theorem~\ref{thm:ge_aac}. 

\begin{thm}[]
	\label{thm:ge_aac}
	Suppose that $\mathcal{F}$ is the maximal frequent item-set in the first-scheduled request, $\hat{t}$ is the minimum cost for channel switching, and $\overline{t_{\text{wait}}}$ is the average waiting time from tuning into the channel to receiving the first required data item for a request, the access time for a request $q_i$ can be expressed as
	\begin{equation}\label{eq2:acc_request}
	acc(q_i) = \left\{
	\begin{array}{cl}
	\overline{t_{\text{wait}}}+|q_i|+\sigma_1\hat{t}+\sigma_2, & \mbox{if $(q_i\subseteq\mathcal{F})\vee (q_i\cap\mathcal{F}=\emptyset)$} \\
	\overline{t_{\text{wait}}}+|q_i\cap\mathcal{F}|+|q_i\setminus\mathcal{F}|+\sigma_1\hat{t}+\sigma_2, & \mbox{otherwise}
	\end{array} \right.
	\end{equation}
	where $\sigma_1$ is the frequency of channel switching and $\sigma_2$ is the frequency of occupied slot (empty node in the FP*-tree) skipping.
\end{thm}
\begin{proof}	
	With the use of index channel in FPBS, the average waiting time can be reduced efficiently. If $q_i\subseteq\mathcal{F}$, it means that all the required data items for $q_i$ can be obtained before the end of broadcasting all the data items in $\mathcal{F}$. In such a case, the access time for $q_i$ will be $\overline{t_{\text{wait}}}+|q_i|+\sigma_1\hat{t}+\sigma_2$, where $|q_i|+\sigma_1\hat{t}+\sigma_2\leq|\mathcal{F}|$. If $q_i\cap\mathcal{F}=\emptyset$ (is equivalent to $|q_i\cap\mathcal{F}|=0$), it means that $q_i$ and $\mathcal{F}$ are two disjoint sets. In this case, the data items requested by $q_i$ only can be allocated after the first-scheduled maximal frequent item-set, so the access time for $q_i$ will be $\overline{t_{\text{wait}}}+|\mathcal{F}|+|q_i|+\sigma_1\hat{t}+\sigma_2$. However, the time $|\mathcal{F}|$ can be merged into the average waiting time $\overline{t_{\text{wait}}}$ until accessing the first data item requested by $q_i$.
	Otherwise, for the case of $|q_i\setminus\mathcal{F}|>0$, $q_i$ and $\mathcal{F}$ are two partially overlapping. It means that some required data items for $q_i$ will be scheduled after $\mathcal{F}$. Hence, the access time for $q_i$ will be $\overline{t_{\text{wait}}}+|q_i\cap\mathcal{F}|+|q_i\setminus\mathcal{F}|+\sigma_1\hat{t}+\sigma_2$, where $|q_i\cap\mathcal{F}|+|q_i\setminus\mathcal{F}|+\sigma_1\hat{t}+\sigma_2\geq|F|$. 
\end{proof}

After discussing the general case of access time, we also discuss the worst case in following Theorem~\ref{thm:worse_aac}.
\begin{thm}[]
	\label{thm:worse_aac}
	Suppose all the notations are defined as above. The worst case of access time will be
	\begin{equation}\label{eq1:tree_height}
	acc_{\text{worst}} = \overline{t_{\text{wait}}}+|\bigcup_{i=1}^{n} q_i|. \nonumber
	\end{equation}
\end{thm}
\begin{proof}
	In general, the worse case is the scenario that a client access the channels from the first time slot to the last time slot. In other words, the worse access time of FPBS will be the height of the FP*-tree. According to the design of FPBS approach, the accelerating branches of FP*-tree is impossible to be longer than the backbone of FP*-tree. Hence, the height of the FP*-tree $\mathcal{H_T}$ will be the height of the backbone, $|\bigcup_{i=1}^{n} q_i|$. In practice, each client tunes in channel at random time slot, so the access time in worst case $acc_{\text{worst}}$ will be $\overline{t_{\text{wait}}}+|\bigcup_{i=1}^{n} q_i|$.
\end{proof}

In FPBS, each data item is not replicated in the FP*-tree's backbone and $|\bigcup_{i=1}^{n} q_i|$. In this work, we focus on minimizing the average access time and the proposed FPBS approach can effectively shorten the access time of each request using the accelerating branches. In \eqref{eq2:acc_request}, the terms $|q_i\cap\mathcal{F}|$ and $|q_i\setminus\mathcal{F}|$ are uncertain since the relation between request $q_i$ and the maximal frequent item-set $F$ is unpredictable. Hence, FPBS focus on minimizing the frequencies of channel switching or occupied slot (empty node in the FP*-tree) skipping, such as $\sigma_1$ and $\sigma_2$ in \eqref{eq2:acc_request}. This problem is solved by FP*-tree using the accelerating branches in our proposed approach. In other words, FPBS is proposed for effectively make the upper bound of access time be tighter. Thus, the worst case becomes a very rare occurrence.

\section{Simulation Results}
\label{simulation}
We validate and discuss the performance of FPBS in terms of average access time by running the experimental simulations in different scenarios. The unit of time is a time slot. All the simulations are written in C++ and executed on a Windows 7 server which is equipped with an Intel (R) Core (TM) i7-3770 CPU @ 3.4GHZ and 12G RAM. We use Quandl databases~\cite{QuandlWIKI} to extract the U.S stock prices and then use the obtained stock dataset as the input of our simulation.

We assume that the maximum number of channels is 10 ($|C|=2,3,\dots,10$) in the simulation. Therefore, we assume that one of the channels is the uplink for receiving the request and the remaining 10 channels are used as the downlink broadcasting channel. The detailed parameters of our simulations are shown in Table~\ref{simulation_parameters}.

In the simulations, FPBS is conducted in online and offline modes. In the online mode, the system will use a buffer to keep the information of queries and request data items. When the buffer becomes full, the system will start to schedule data into the broadcasting channels. The scheduled data items will be removed from the buffer and new user demands are continuously coming in the buffer. It means that the FP*-tree and schedule may change during the simulation. Conversely, we assume that the system in the offline mode schedules the data after storing all the requested information in the buffer.
\begin{table}[!ht]
	\renewcommand{\arraystretch}{1.1}
	\caption{Simulation Parameters}
	\label{simulation_parameters}
	\centering
	\small
	\begin{tabular}{p{3.4cm}M{1cm}M{3.1cm}}
		\hline
		\textbf{Parameter} & \textbf{Default Value} & \textbf{Range (type)}\\
		\hline
		Size of dataset, $|D|$ & 500 & 100, 300, 500, 700, 900\\
		Number of users & 5000 & --\\
		\begin{tabular}{@{}l} Maximum number of \\requested data items, $q_{max}$
		\end{tabular}
		& 10 & 2, 4, 6, 8, 10\\
		Number of downlink broadcast channels, $|C|$ & 6 & 2, 3, $\dots$, 10\\
		Size of buffer & 3000 & 500, 1000, $\dots$, 4500\\
		\hline
	\end{tabular}
\end{table}
\begin{figure*}[!t]
	\centering
	\subfigure[$|C|=3$]{
		\label{fig:datasize:c3} 
		\includegraphics[width=0.325 \textwidth]{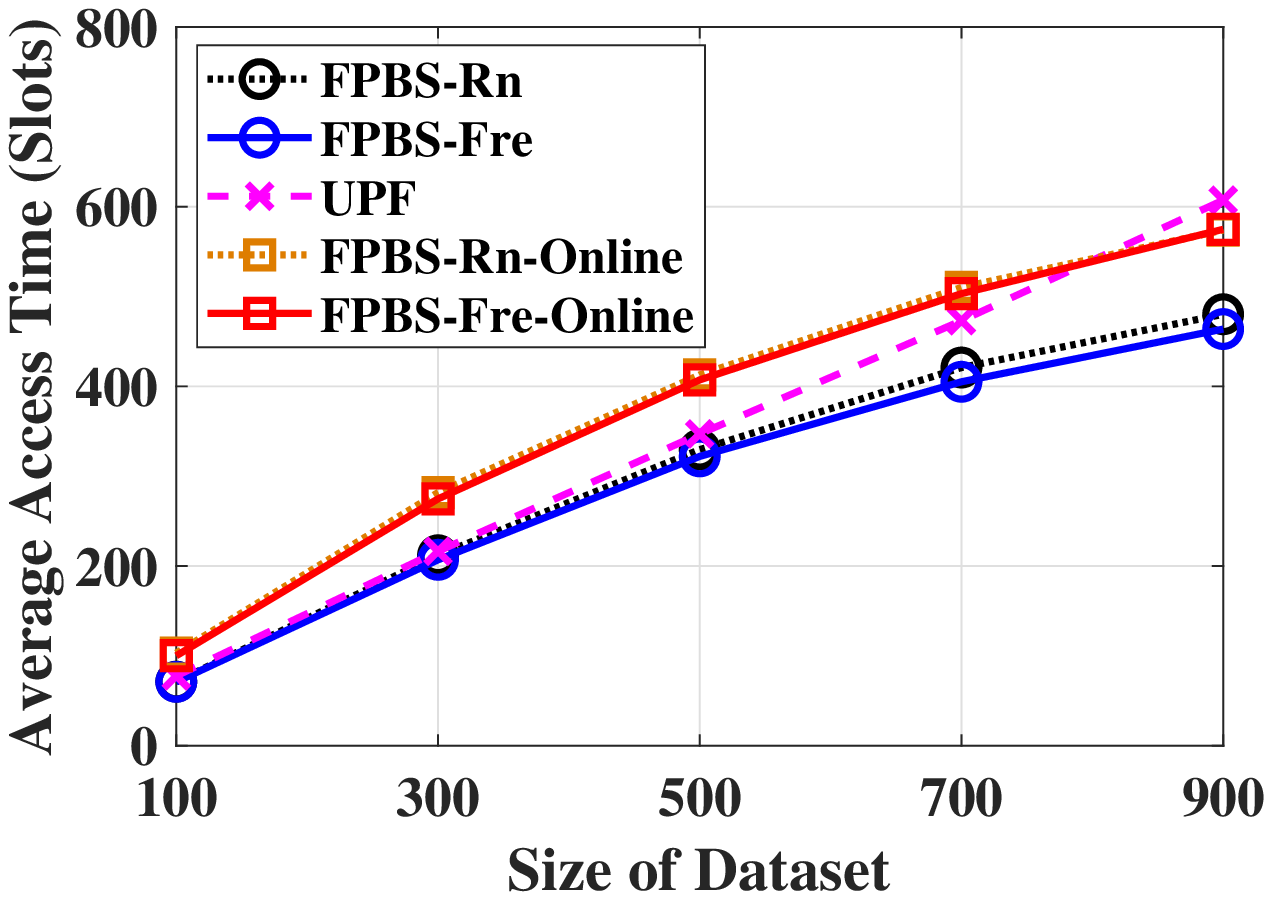}}
	\subfigure[$|C|=6$]{
		\label{fig:datasize:c6} 
		\includegraphics[width=0.325 \textwidth]{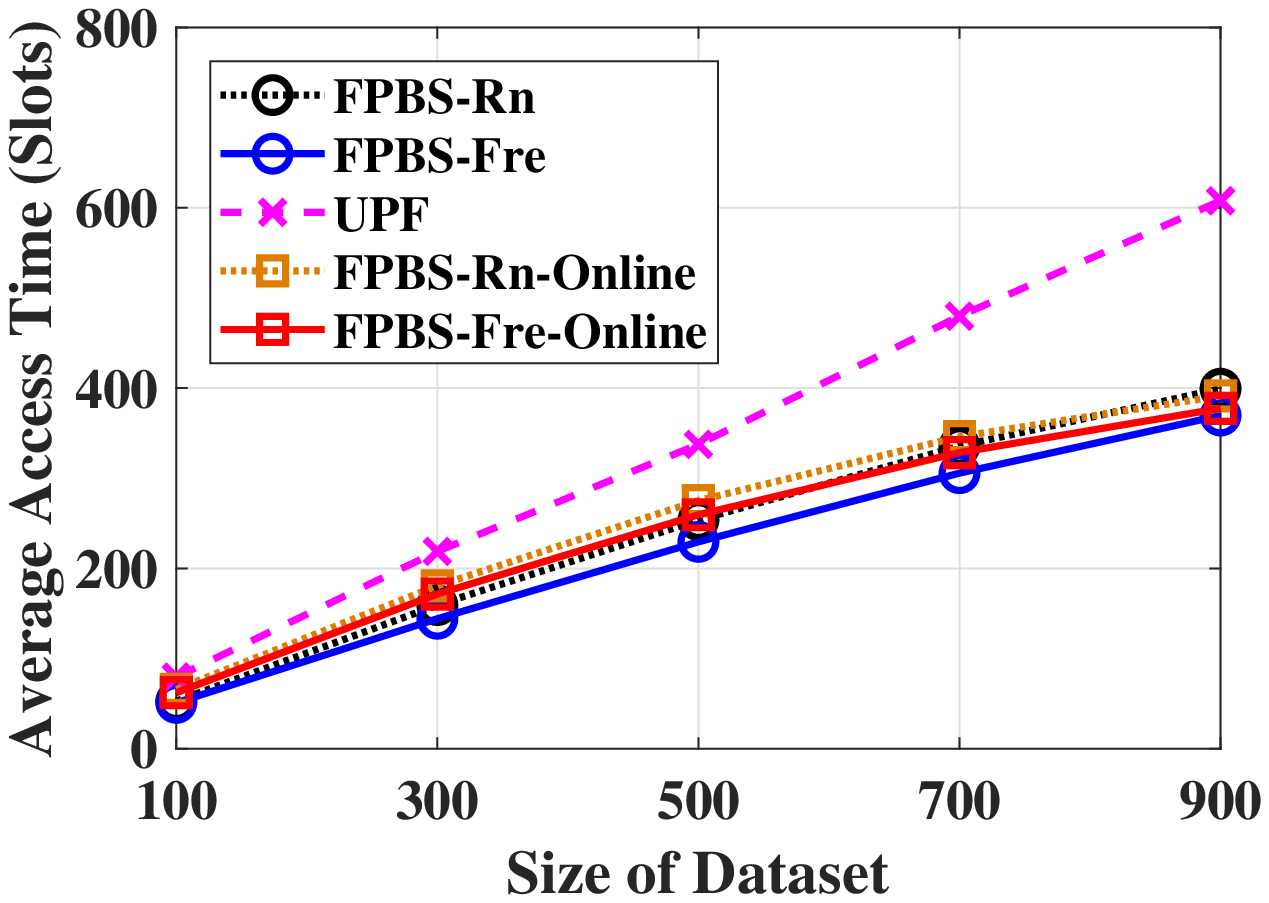}}
	\subfigure[$|C|=9$]{
		\label{fig:datasize:c9} 
		\includegraphics[width=0.325 \textwidth]{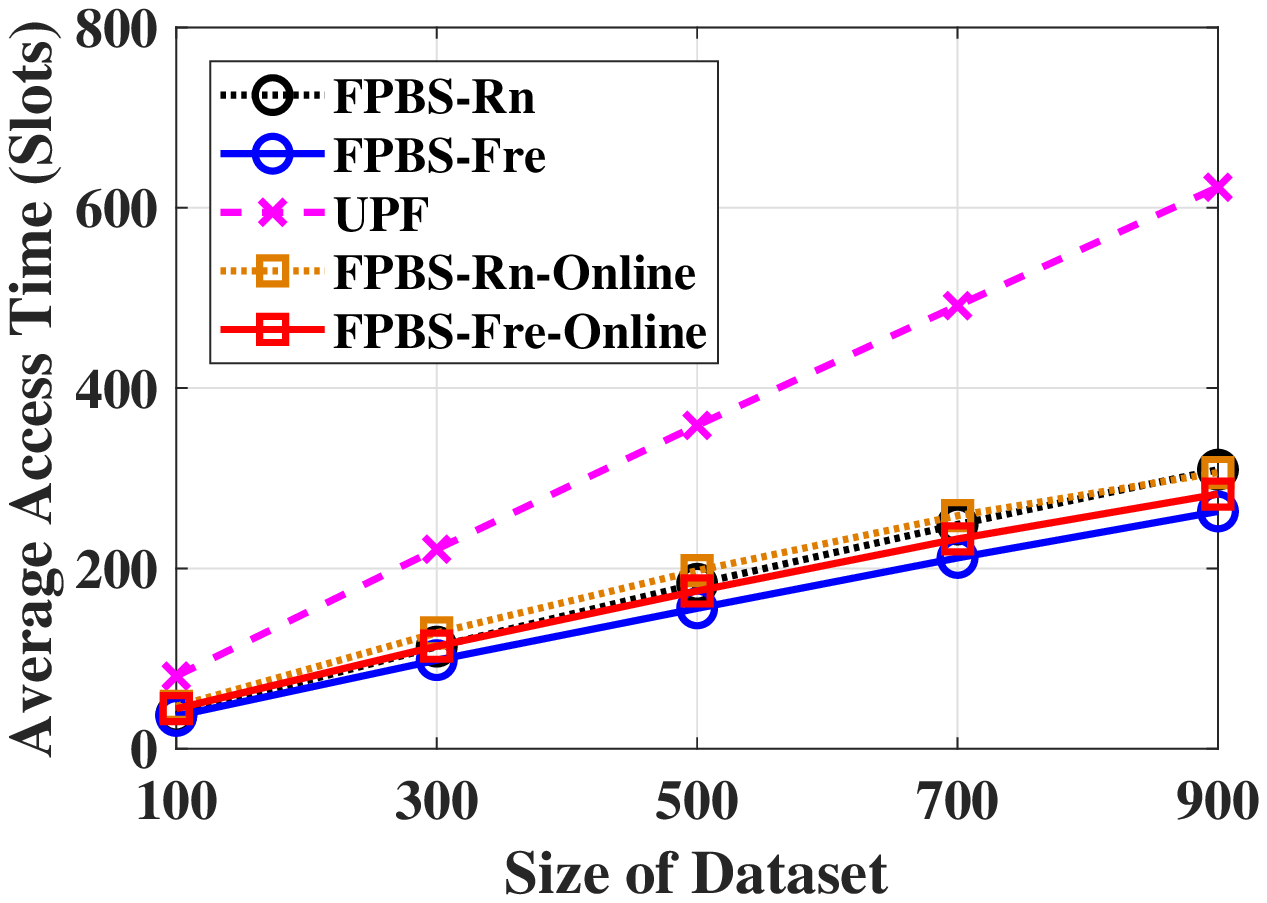}}
	\caption{Effect of the different sizes of dataset with different number of channels: \subref{fig:datasize:c3} $|C|=3$, \subref{fig:datasize:c6} $|C|=6$, and \subref{fig:datasize:c9} $|C|=9$
	}
	\label{fig:datasize} 
	\vspace{-15pt}
\end{figure*}

Note that there are two selecting strategies during scheduling process of FPBS, Request-Number-First and Frequency-First. Request-Number-First strategy is to select the query according to the length of its requested data items first and then selecting the query according to its average access frequency if multiple queries request same number of data items. Frequency-First strategy is to select the query according to its average access frequency first and then select the query according to the length of its requested data items if multiple queries have the same average access frequency. Hence, we discuss the above two strategies in online and offline modes respectively.

\begin{figure*}[!t]
	\centering
	\subfigure[$|D|=100$]{
		\label{fig:channels:100} 
		\includegraphics[width=0.325 \textwidth]{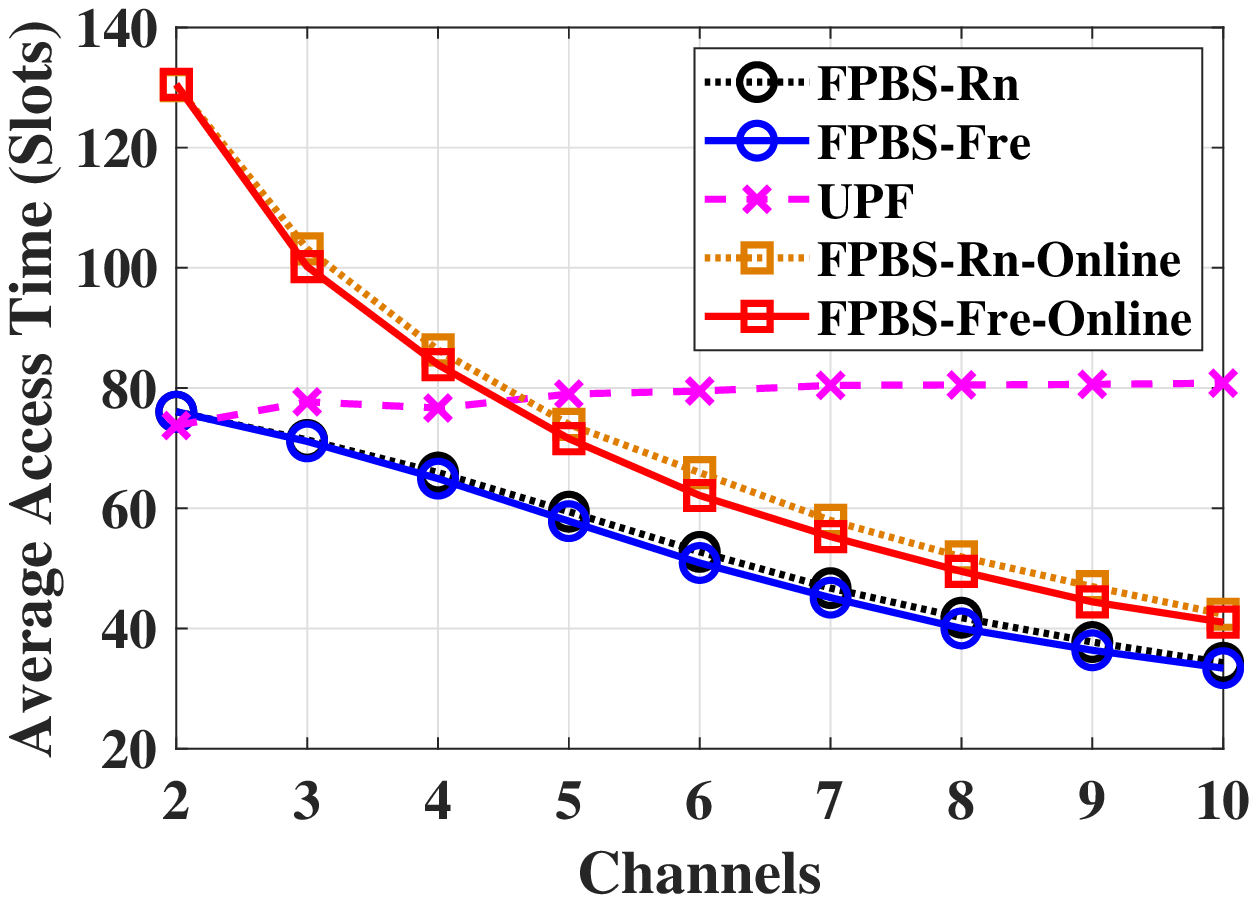}}
	\subfigure[$|D|=300$]{
		\label{fig:channels:300} 
		\includegraphics[width=0.325 \textwidth]{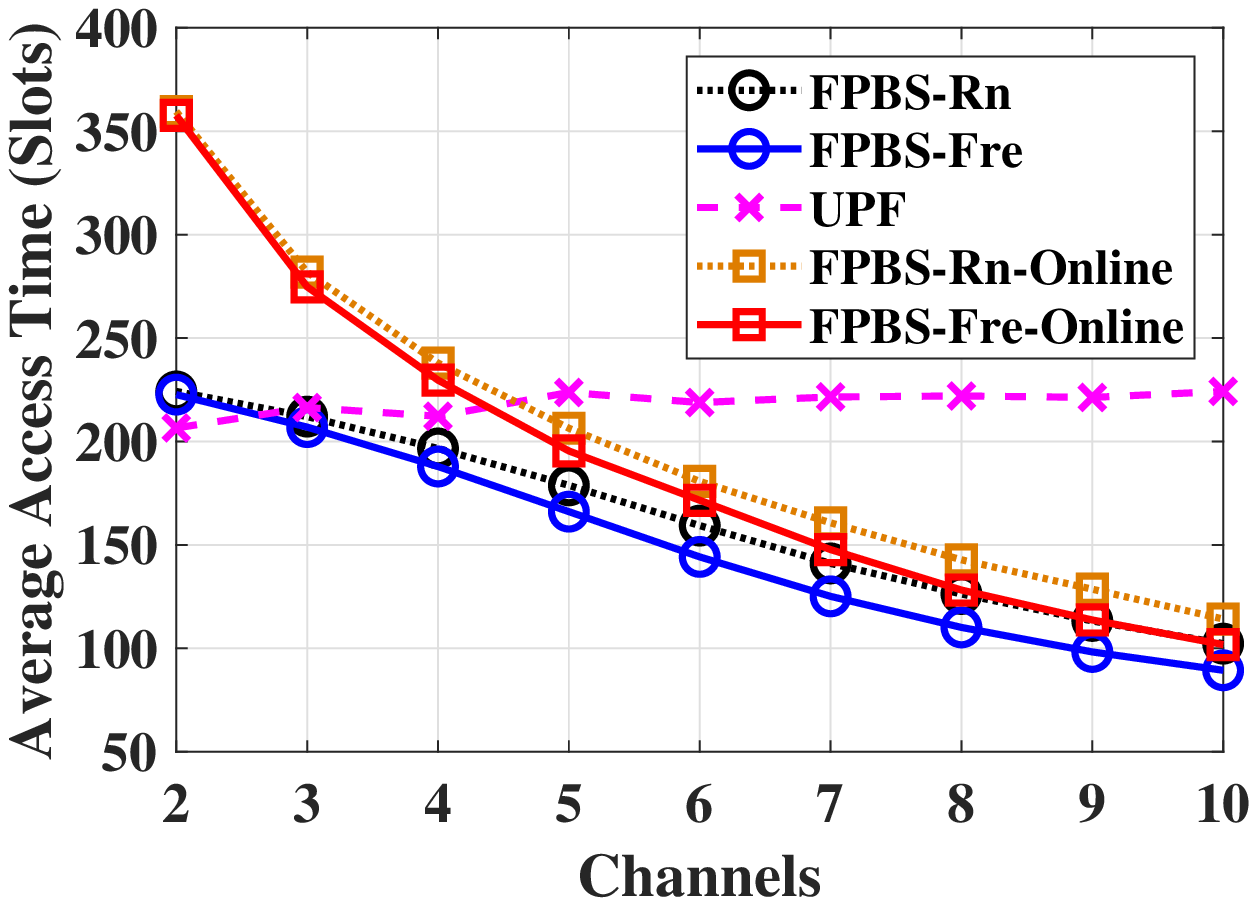}}
	\subfigure[$|D|=500$]{
		\label{fig:channels:500} 
		\includegraphics[width=0.325 \textwidth]{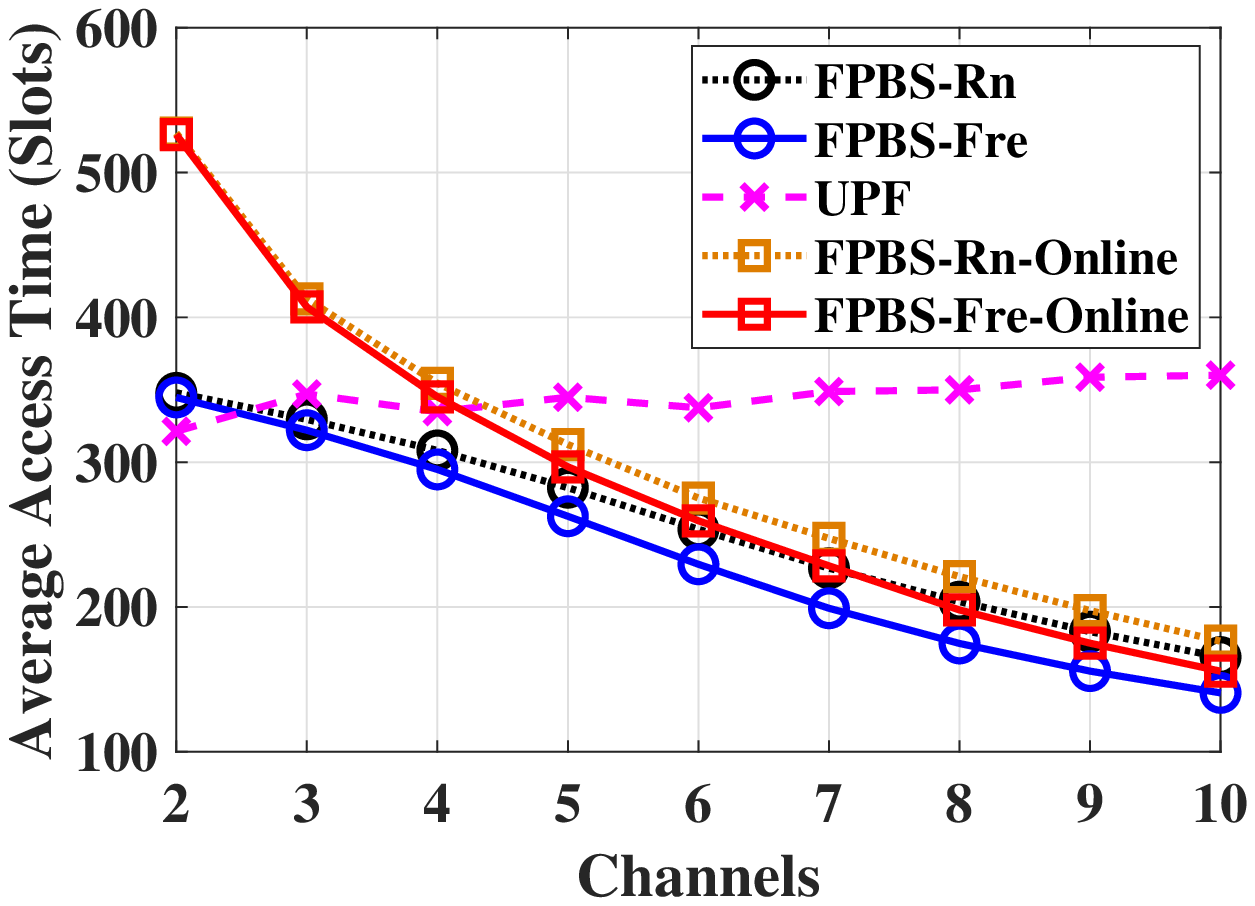}}\\ \vspace{-5pt}
	\subfigure[$|D|=700$]{
		\label{fig:channels:700} 
		\includegraphics[width=0.325 \textwidth]{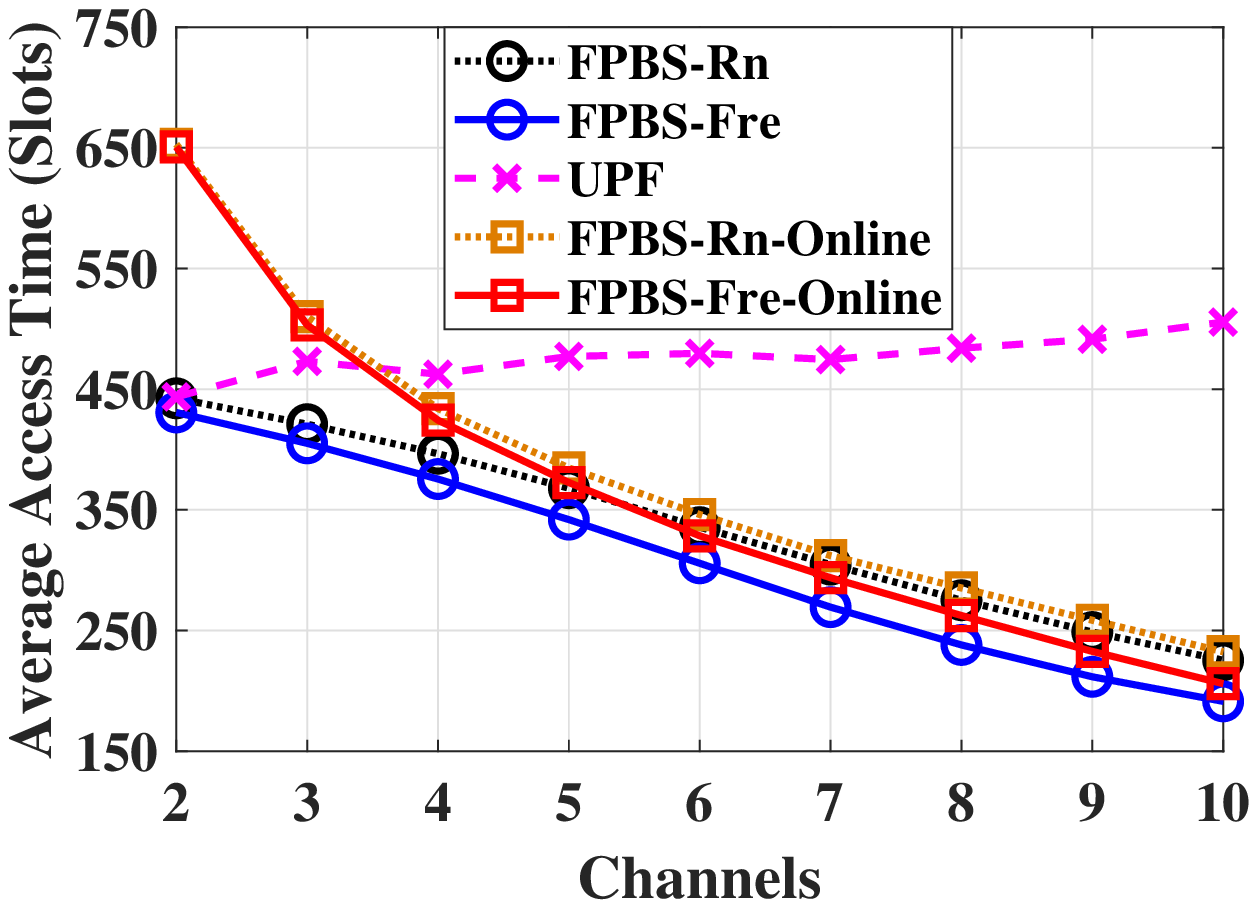}}
	\subfigure[$|D|=900$]{
		\label{fig:channels:900} 
		\includegraphics[width=0.325 \textwidth]{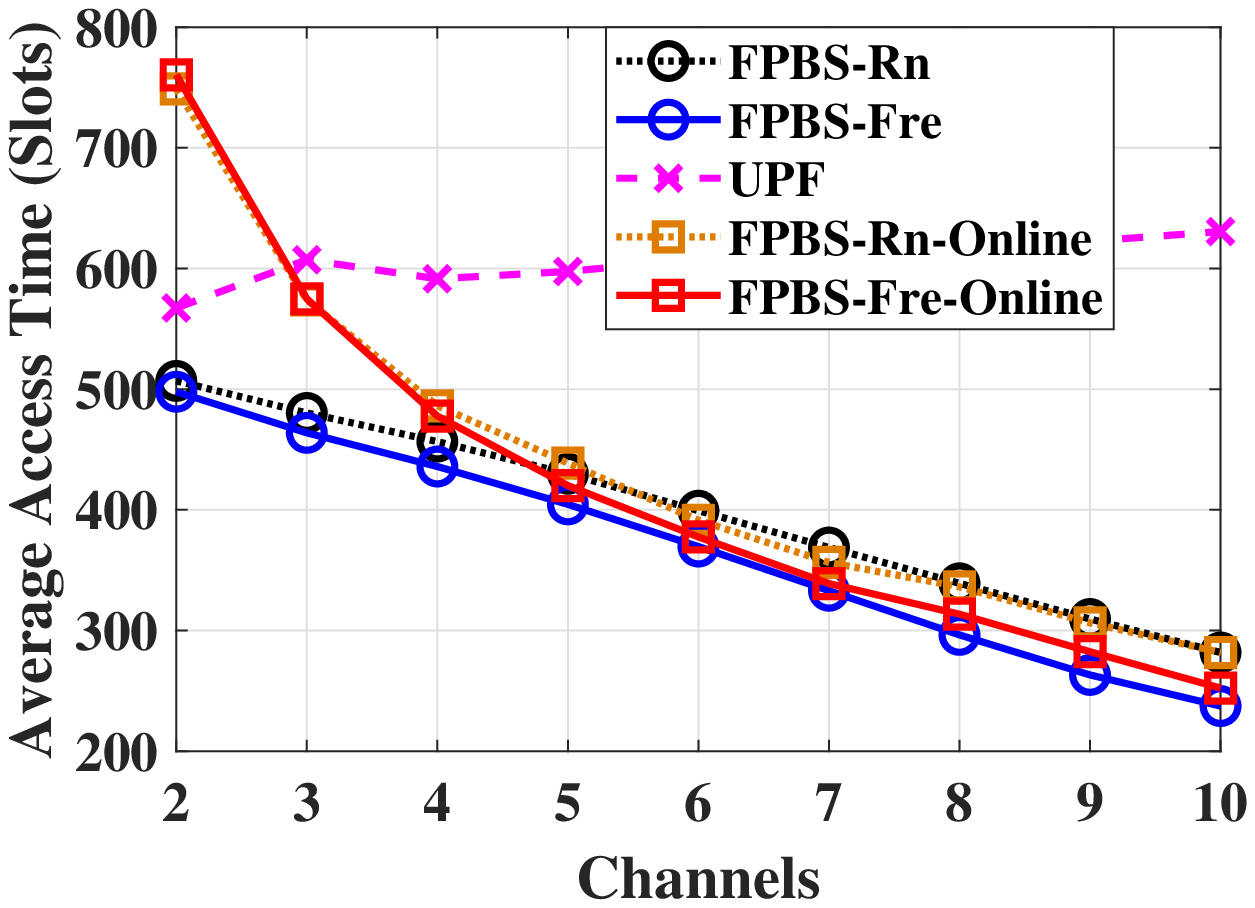}}
	\vspace{-5pt}
	\caption{Effect of the different number of channels with different sizes of dataset: \subref{fig:channels:100} $|D|=100$, \subref{fig:channels:300} $|D|=300$, \subref{fig:channels:500} $|D|=500$, \subref{fig:channels:700} $|D|=700$, and
		\subref{fig:channels:900} $|D|=900$
	}
	\label{fig:channels} 
\end{figure*}
\begin{figure*}[!t]
	\centering
	\subfigure[$|D|=100$]{
		\label{fig:num:request:data:100} 
		\includegraphics[width=0.325 \textwidth]{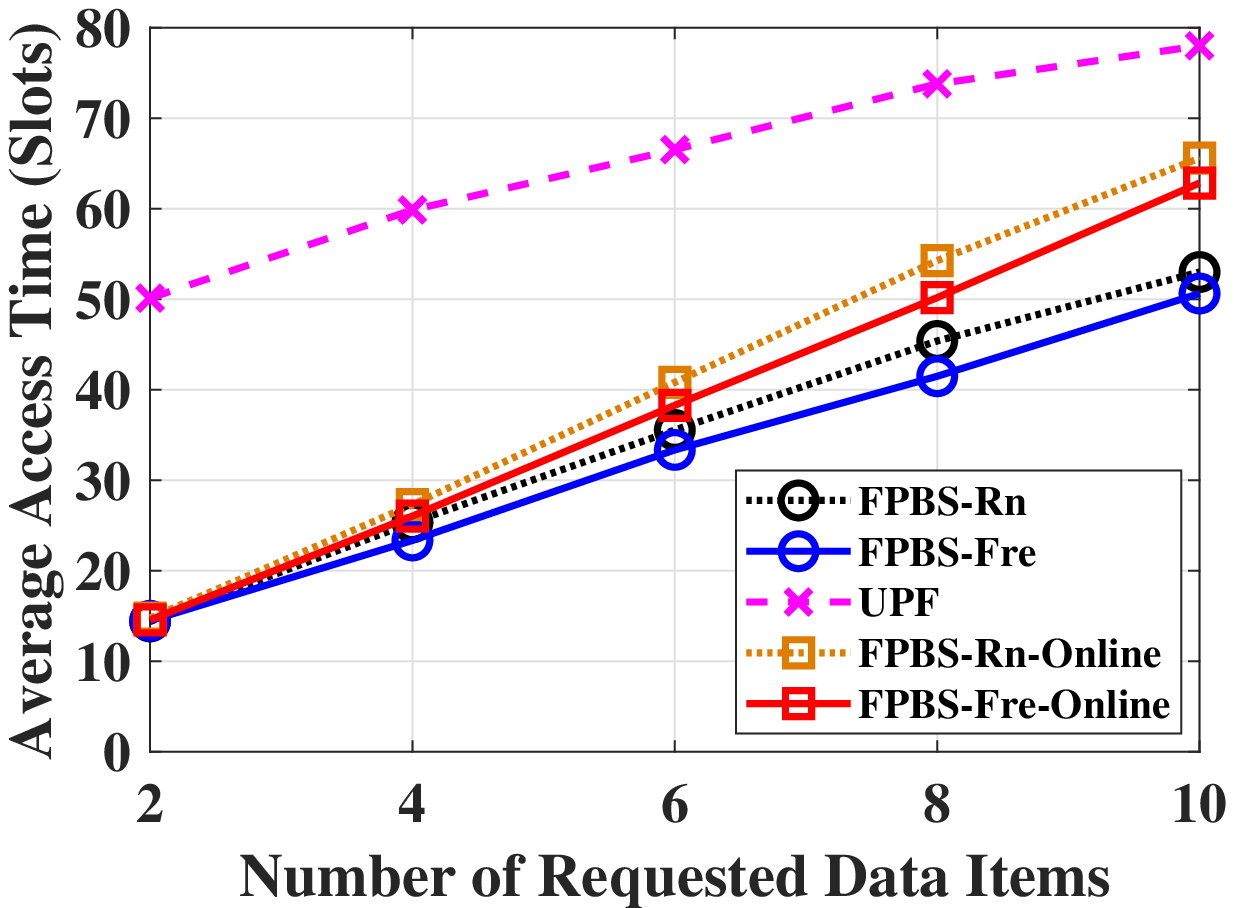}}
	\subfigure[$|D|=300$]{
		\label{fig:num:request:data:300} 
		\includegraphics[width=0.325 \textwidth]{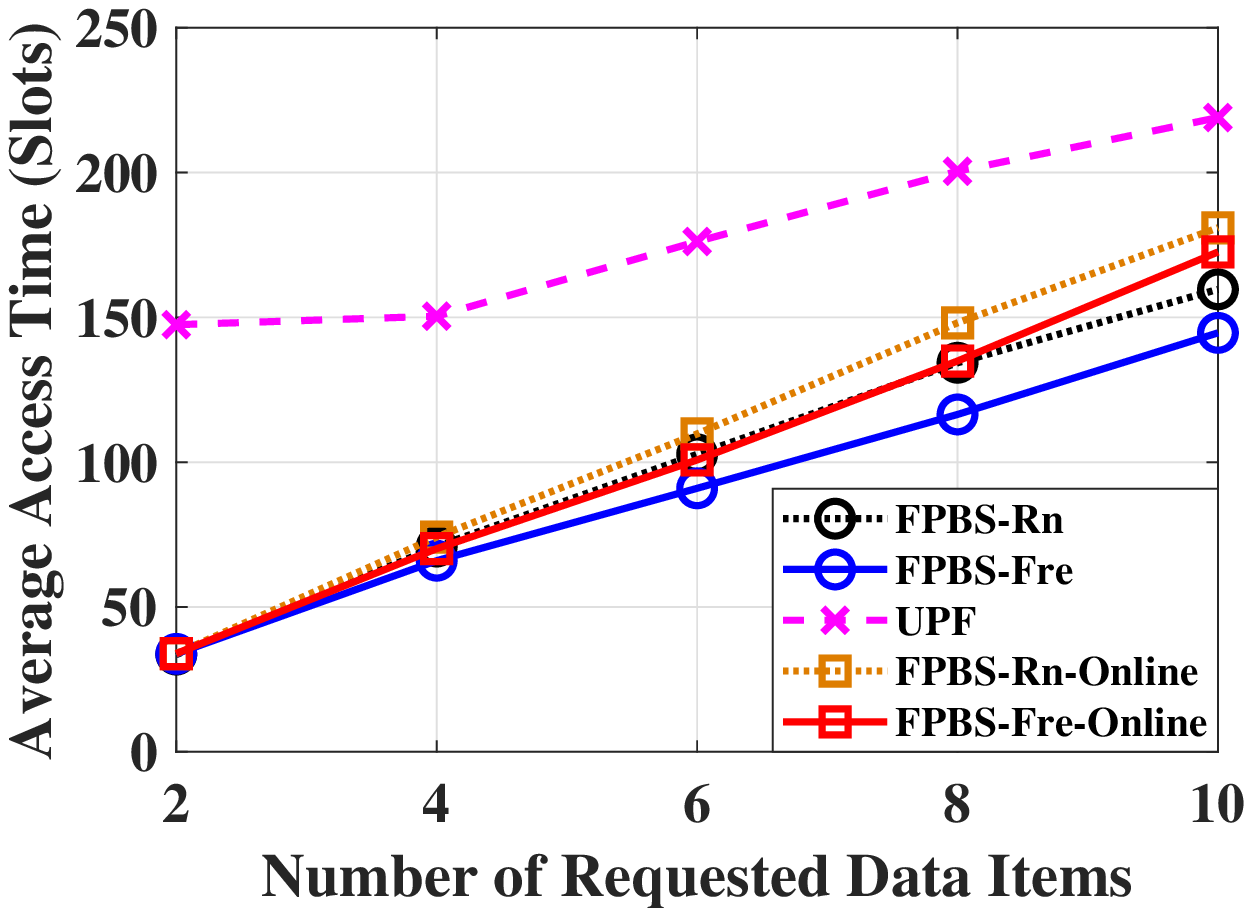}}
	\subfigure[$|D|=500$]{
		\label{fig:num:request:data:500} 
		\includegraphics[width=0.325 \textwidth]{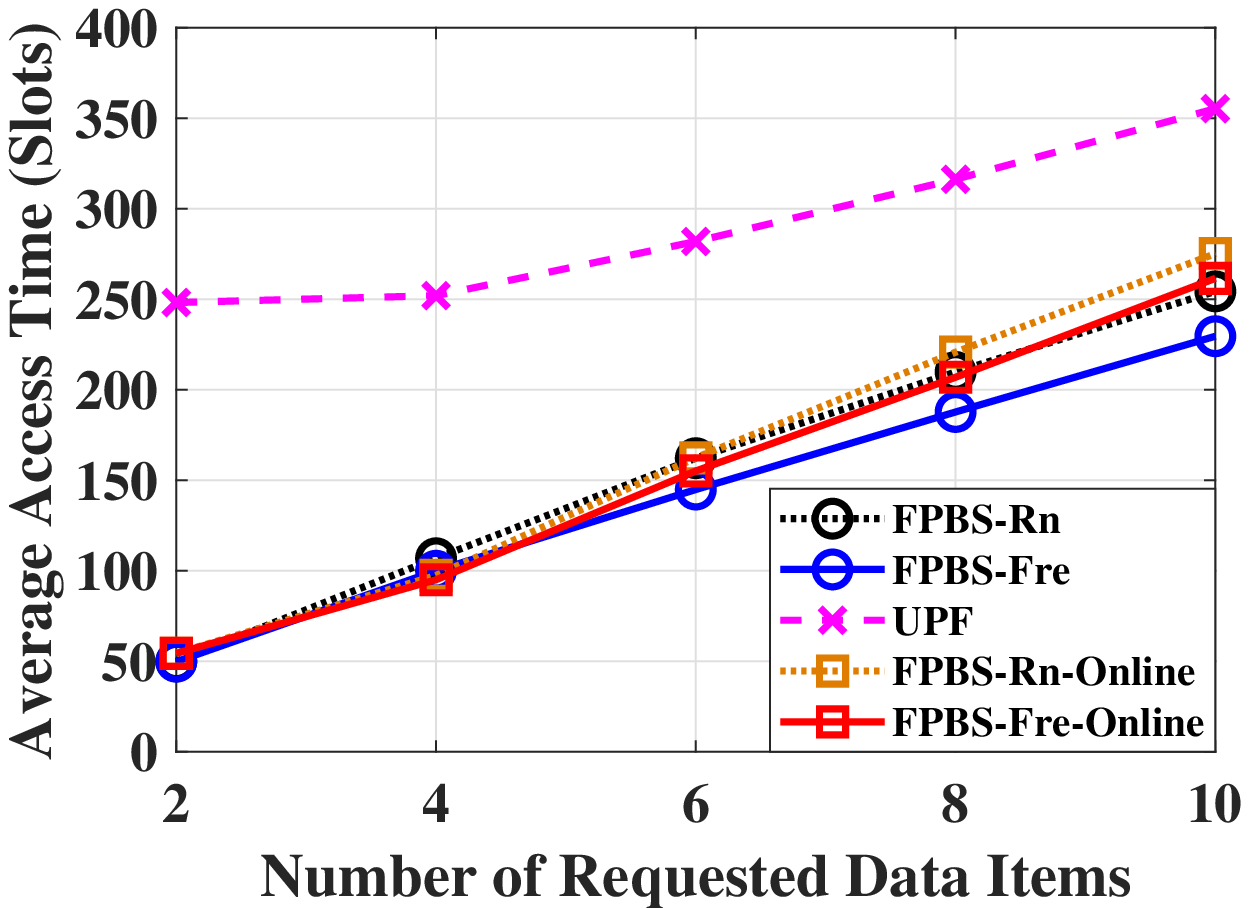}}\\ 
	\vspace{-5pt}
	\subfigure[$|D|=700$]{
		\label{fig:num:request:data:700} 
		\includegraphics[width=0.325 \textwidth]{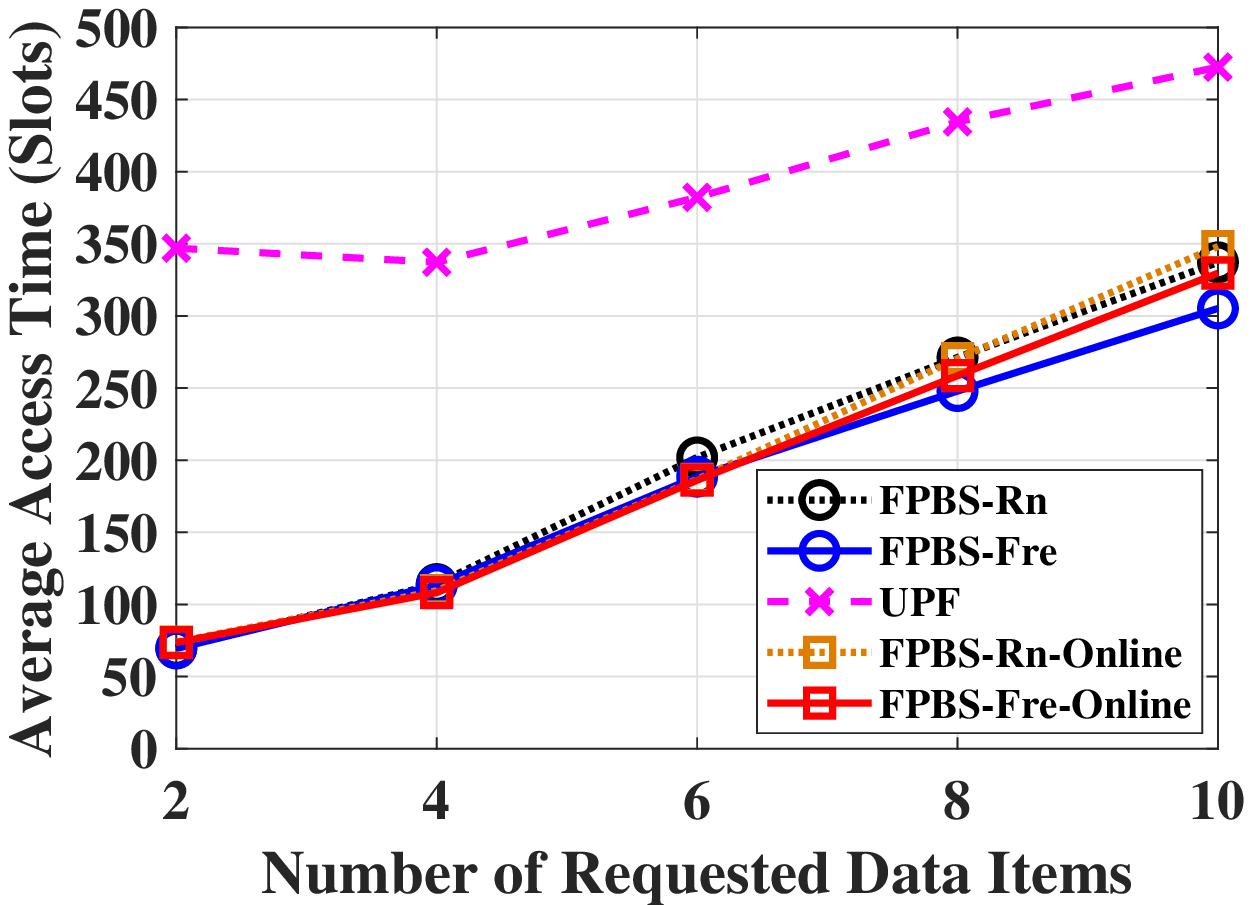}}
	\subfigure[$|D|=900$]{
		\label{fig:num:request:data:900} 
		\includegraphics[width=0.325 \textwidth]{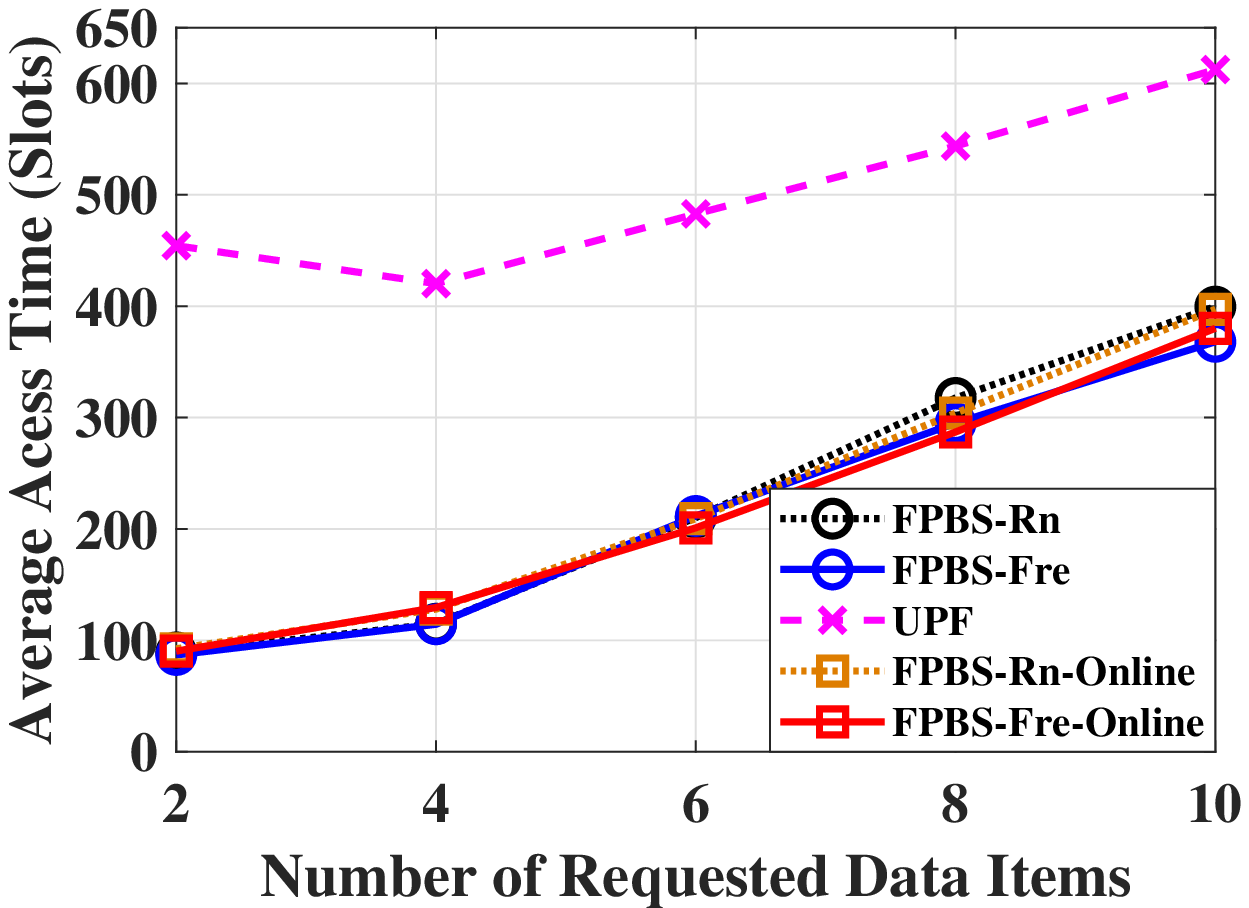}}
	\caption{Effect of the different number of requested data items with different sizes of dataset: \subref{fig:num:request:data:100} $|D|=100$, \subref{fig:num:request:data:300} $|D|=300$, \subref{fig:num:request:data:500} $|D|=500$, \subref{fig:num:request:data:700} $|D|=700$, and
		\subref{fig:num:request:data:900} $|D|=900$
	}
	\label{fig:num:request:data} 
\end{figure*}

To the best of our knowledge, none of existing works model the optimal performance of the multi-item request scheduling simultaneously considering the channel switching and dependencies between different requests over multi-channel dissemination environments. Only~\cite{He2015118} provides a heuristic algorithm, UPF, to discuss the similar problem. This is the reason that we choose UPF as the comparative baseline in the simulations.

\subsection{Size of Dataset}
In the first simulation, we discuss the performance of FPBS with different sizes of dataset in terms of average access time. Note that the size of dataset indicates the number of different data items stored in the dataset. Fig.~\ref{fig:datasize} shows the results in three different cases if the number of channels $|C|=3$, $|C|=6$, and $|C|=9$, respectively. In the $|C|=3$ channels environment, as shown in Fig.~\ref{fig:datasize:c3}, UPF can outperform the online FPBS approaches, FPBS-Fre-Online and FPBS-Rn-Online, if the size of dataset, $|D|$, is smaller than 800. The offline FPBS, FPBS-Fre and FPBS-Rn, can always have a better performance than UPF does in all different sizes of dataset.

The results depicted from Fig.~\ref{fig:datasize:c3} to Fig.~\ref{fig:datasize:c9} show that UPF has similar performances in different number of channels environments and the trends of UPF's average access time are always linear increasing. According to the results in Fig.~\ref{fig:datasize:c6} and Fig.~\ref{fig:datasize:c9}, we can know that both of online and offline FPBS approaches can outperform UPF in different sizes of datasets when $|C|\geq 6$. Additionally, the Frequency-First strategy, FPBS-Fre, always has the best performance in different scenarios.

\subsection{Number of Channels}
In this part, we discuss the performance of FPBS in different scenarios that the number of broadcasting channels is set from 2 to 20 and the results are shown in Fig.~\ref{fig:channels}. The results indicate the existing method, UPF, is not suitable to multiple channel ($C\geq 4$) broadcasting environments and UPF cannot dynamically schedule data items with the consideration of each user's requests. That is to say, in comparison with the proposed approach, UPF can not utilize these channels if $C\geq 4$. Fig.~\ref{fig:channels:100} and Fig.~\ref{fig:channels:300} show that UPF has a stable performance in the broadcasting environments with different number of channels when the size of dataset is small ($|D|\leq 300$). Conversely, the results from Fig.~\ref{fig:channels:500} to Fig.~\ref{fig:channels:900} show that the average access time of UPF is unstable and becomes a slightly increasing trend when the size of dataset becomes large ($|D|\geq 500$). 
The possible reason for this result is that UPF aims to minimize the request miss rate, not the average access time. There may be a trade-off between minimizing the request miss rate and the average access time.

\begin{figure*}[!t]
	\centering
	\subfigure[$|C|=3$]{
		\label{fig:buffersize:c3} 
		\includegraphics[width=0.325 \textwidth]{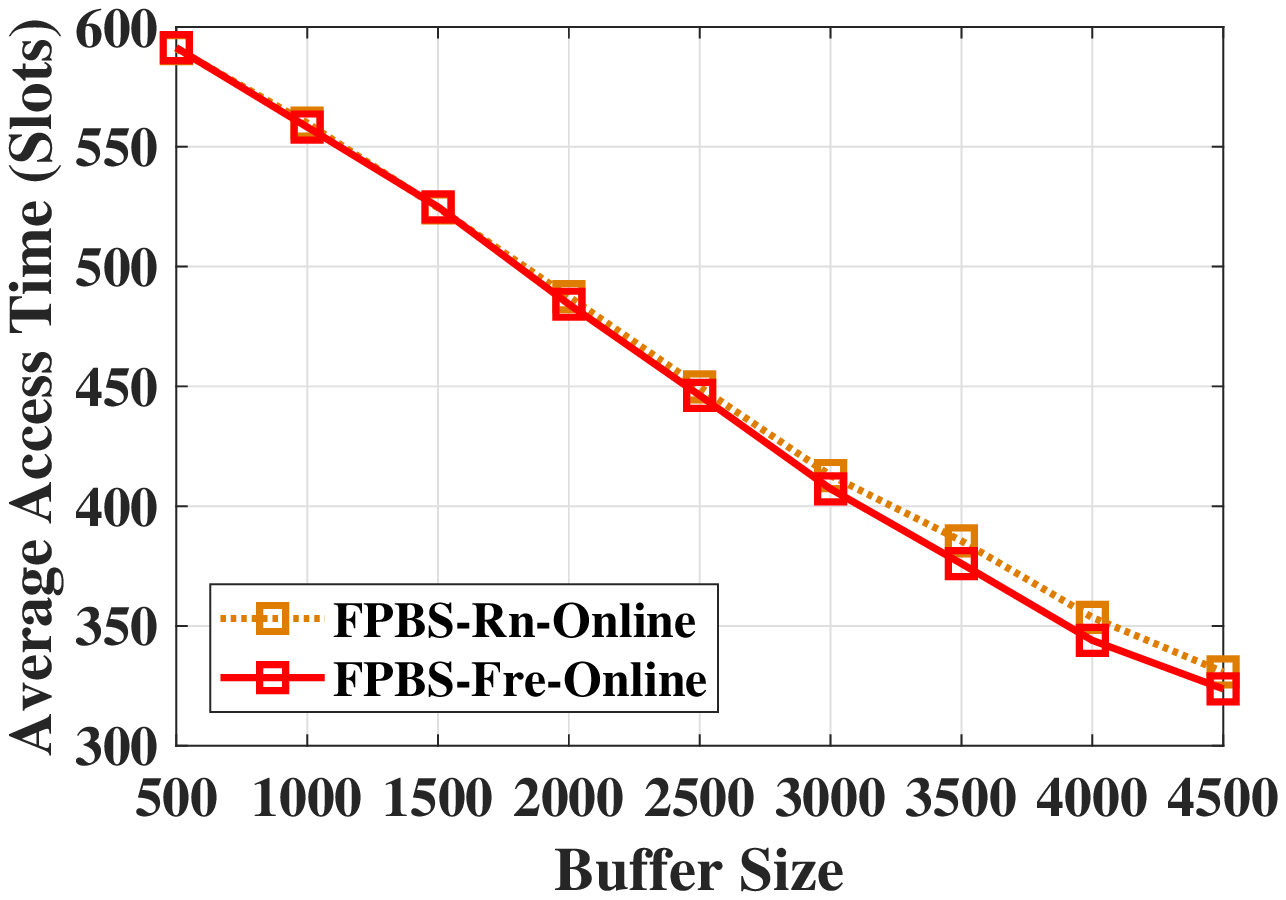}}
	\subfigure[$|C|=6$]{
		\label{fig:buffersize:c6} 
		\includegraphics[width=0.325 \textwidth]{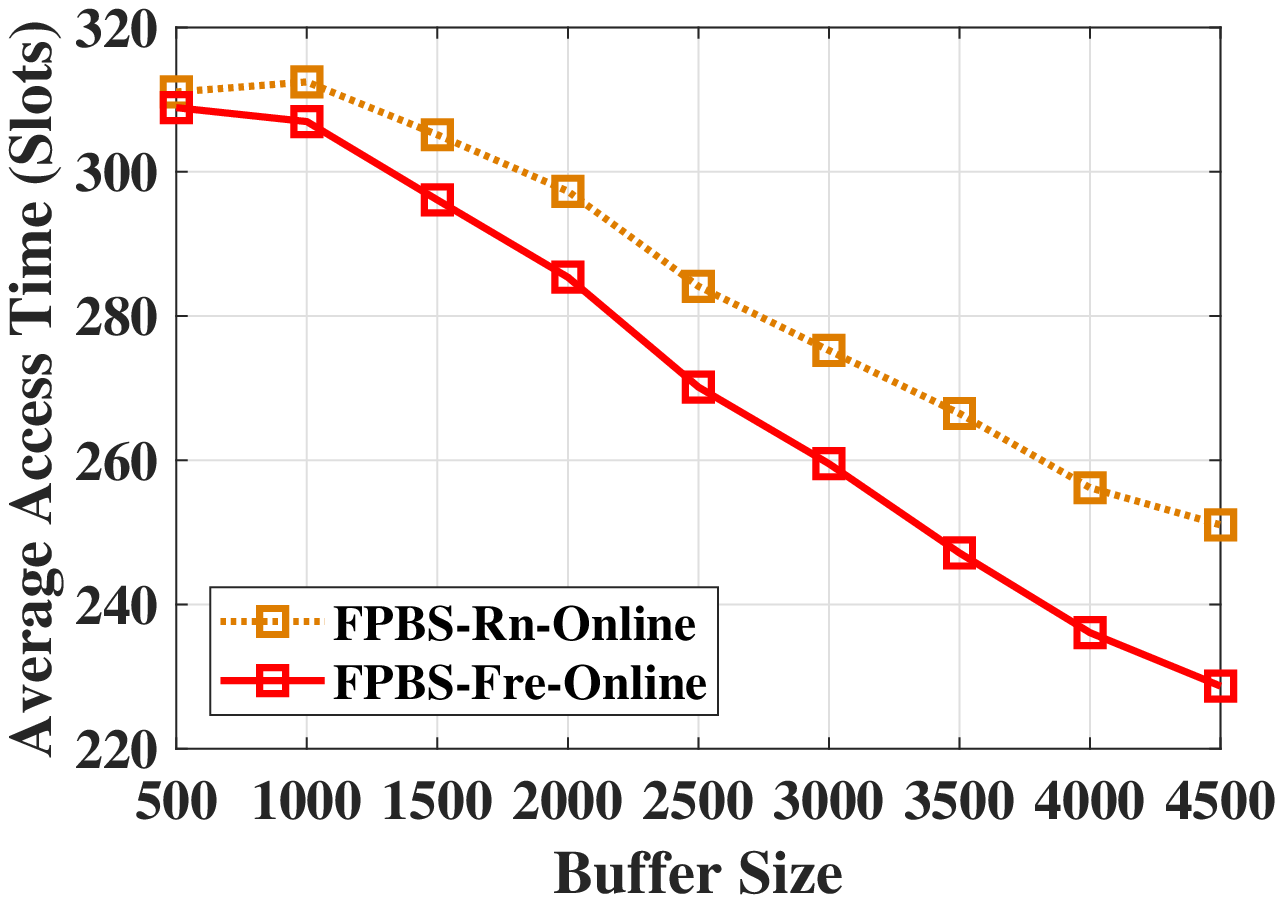}}
	\subfigure[$|C|=9$]{
		\label{fig:buffersize:c9} 
		\includegraphics[width=0.325 \textwidth]{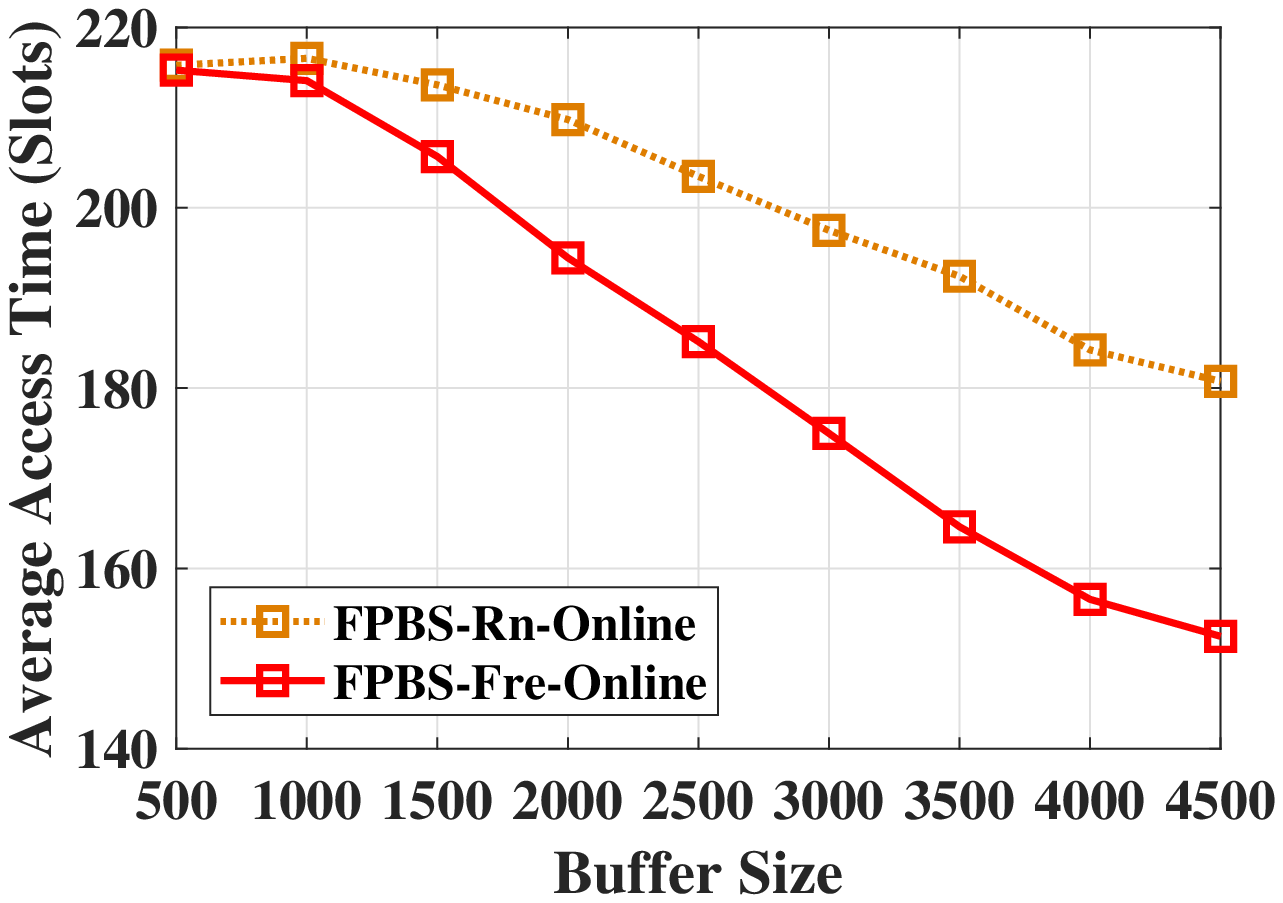}}
	\vspace{-5pt}
	\caption{Effect of the different size of buffer with different number of channels: \subref{fig:buffersize:c3} $|C|=3$, \subref{fig:buffersize:c6} $|C|=6$, and \subref{fig:buffersize:c9} $|C|=9$
	}
	\label{fig:buffersize} 
\end{figure*}

Fig.~\ref{fig:channels:100} and Fig.~\ref{fig:channels:300} shows the results of each approach in small dataset. FPBS in the offline mode, FPBS-Fre and FPBS-Rn, can have a better performance since the system consider all the requests while constructing the FP*-tree. According to the results in Fig.~\ref{fig:channels:500}, Fig.~\ref{fig:channels:700}, and Fig.~\ref{fig:channels:900}, the Frequency-First strategies, FPBS-Fre and FPBS-Fre-Online, have better performances than the Request-Number-First strategies, FPBS-Rn and FPBS-Rn-Online when the size of dataset becomes large ($|D|\geq 500$).

\subsection{Number of Requested Data Items}
If the number of requested data items becomes larger, the possibility of data dependency between each query becomes higher.
In this subsection, we consider the effect of the different number of requested data items on the average access time. As shown in Fig.~\ref{fig:num:request:data}, one can observe that all the FPBS-based approaches can outperform UPF when the maximum number of requested data items $q_{max}$ is smaller than 11. When $q_{max}$ is 2, all the FPBS-based approaches have similar performances on the average access time. As the value of $q_{max}$ increases, the average access time in all the FPBS-based approaches also increases linearly.

According to the result in Fig.~\ref{fig:num:request:data}, we can know that the Frequency-First strategies are better than the Request-Number-First strategies since the performances of FPBS-Fre and FPBS-Fre-Online are more smoothly increasing than the performances of FPBS-Rn and FPBS-Rn-Online. In addition, FPBS-Fre can has the best performance and its trend is almost parallel to the trend of UPS's performance.

\subsection{Buffer Size}
In the last simulation, we discuss the effect of the different size of buffer on the average access time for comparing two proposed online approaches, FPBS-Rn-Online and FPBS-Fre-Online. We also consider the trend of performance in some scenarios that the number of channel is respectively set to 3, 6, and 9.

The result in Fig.~\ref{fig:buffersize} indicates that both FPBS-Rn-Online and FPBS-Fre-Online can have shorter average access time as the size of buffer increases. In an environment providing small number ($C=3$) of channels, as shown as Fig.~\ref{fig:buffersize:c3}, FPBS-Fre-Online can has a slightly better performance than FPBS-Rn-Online does when the buffer can store more than 2500 data items. The results in Fig.~\ref{fig:buffersize:c6} and Fig.~\ref{fig:buffersize:c9} show that FPBS-Fre-Online is much better than FPBS-Rn-Online with different size of buffer when the number of channels increases ($C\geq 6$).

\subsection{Open Issues}
In this subsection, we summarize some remaining issues (or potential challenges) in on-demand multi-channel data dissemination systems as follows:
\begin{itemize}
	\item \textbf{Hardware constraint:} Although the minimum cost $\hat{t}$ for channel switching is normalized as one time slot in FPBS, it is difficult to implement a broadcasting system that meets this condition due to hardware limitations.
	\item \textbf{Cross-layer system design:} In this paper, we design a server-side data scheduling for serving the multi-item requests. For wireless networks, the time-varying and uncertain nature of wireless channels can be considered in the scheduling. Thus, the server needs a new cross-layer system design to simultaneously access the request information in the application layer and channel information in the physical layer and then schedule data items more efficiently.
\end{itemize}

\section{Conclusion}
\label{conclusion}
In this paper, we investigate and formulate an emerging problem, DBCA, in multi-channel wireless data dissemination environments. We also prove that the DBCA problem is $\mathcal{NP}$-complete. Then, we present a heuristic scheduling approach, FPBS, to avoid data conflicts on multiple broadcasting channels. In FPBS, we use frequent patterns of requested data items to build a FP*-tree for extracting the correlation between each received request. Thus, data conflicts can be avoided. During the construction of FP*-tree's accelerating branch, adding empty nodes at appropriate positions makes the user client have sufficient time to switch the channel for obtaining the required data. We not only analyze that FPBS can be done in polynomial time but also present the upper-bound of access time of a request which is related to size of dataset. According to the simulation results, FPBS is much better than the existing work, UPF, in most of cases.


%



\ifCLASSOPTIONcompsoc
  \section*{Acknowledgments}
\else
  \section*{Acknowledgment}
\fi

This research was supported by the Ministry of Science and Technology, Taiwan under Grant No. MOST 107-2221-E-027-099-MY2, MOST 109-2221-E-027-095-MY3, and MOST 110-2222-E-035-004-MY2.

\ifCLASSOPTIONcaptionsoff
  \newpage
\fi



\bibliographystyle{IEEEtran}
\bibliography{reference}
\end{document}